\documentclass[a4paper, UKenglish, cleveref, autoref, thm-restate]{lipics-v2021}

\hideLIPIcs  

\title{Continuous Defensive Domination Problems} 

\author{Christoph Grüne}{Faculty of Computer Science, RWTH Aachen University, Germany}{gruene@algo.rwth-aachen.de}{https://orcid.org/0000-0002-7789-8870}{}
\author{Tom Janßen}{Faculty of Computer Science, RWTH Aachen University, Germany}{janssen@algo.rwth-aachen.de}{https://orcid.org/0000-0003-4617-3540}{Funded by the German Research Foundation (DFG) – WO 1451/2-1}

\authorrunning{C. Grüne and T. Janßen} 

\Copyright{Christoph Grüne and Tom Janßen} 


\keywords{Facility Location, Defensive Domination, Dominating Set, Computational Complexity} 

\category{} 

\relatedversion{} 

\acknowledgements{We want to thank the anonymous reviewers for their helpful feedback.}

\nolinenumbers 

\EventEditors{John Q. Open and Joan R. Access}
\EventNoEds{2}
\EventLongTitle{42nd Conference on Very Important Topics (CVIT 2016)}
\EventShortTitle{CVIT 2016}
\EventAcronym{CVIT}
\EventYear{2016}
\EventDate{December 24--27, 2016}
\EventLocation{Little Whinging, United Kingdom}
\EventLogo{}
\SeriesVolume{42}
\ArticleNo{23}

\usepackage{mathtools}

\usepackage{tikz}
\usetikzlibrary{positioning,calc,arrows.meta,decorations.pathreplacing,angles,quotes,shapes.geometric}

\newenvironment{numberedclaim}{\begin{claim}}{\end{claim}}

\newcommand{\oball}[2]{\ensuremath{B^<\left(#1, #2\right)}}
\newcommand{\cball}[2]{\ensuremath{B^\leq\left(#1, #2\right)}}

\newcommand{\half}{\ensuremath{\tfrac{1}{2}}}

\newcommand{\mult}[2]{\ensuremath{\operatorname{mult}(#1, #2)}}

\newcommand{\N}{\mathbb{N}}

\newcommand{\NP}{{\sf NP}}

\begin{document}

\maketitle

\begin{abstract}
 The problem {\sc Defensive $\delta$-Covering}, for some covering range $\delta > 0$, is a continuous facility location problem on undirected graphs where all edges have unit length. 
It is a generalization of {\sc Defensive Dominating Set} and {\sc $\delta$-Covering}.
An attack and defense are sets of points, which are on vertices or in the interior of an edge.
A defense counters an attack if there is a matching of the points in the defense to the points in the attack, such that each pair of matched points has a distance of at most $\delta$, and every point in the attack is matched.
The task is, given a graph $G$ and numbers $\ell, k \in \mathbb N$, to find a defense of size at most $\ell$ that counters every possible attack of size at most $k$.

We study the complexity of this problem in various settings.
We show that if the attack is restricted to vertices, the problem is $\Sigma^P_2$-complete for large $\delta$, but if the attack may consist of arbitrary points on the graph, it is $\NP$-complete.
Additionally, we analyze how the complexity changes if the attacks or defenses may be multisets.
If the defense is allowed to be a multiset, the complexity does not change in any case we consider, while if the attack is allowed to be a multiset, the problem often becomes easier.
To show containment in the various complexity classes, we introduce a number of discretization arguments, which show that solutions with a regular structure must always exist.
\end{abstract}
\newpage

\section{Introduction}

The umbrella term ``facility location'' refers to a class of optimization problems where one wants to find locations for facilities and an assignment of clients to these facilities so as to optimize a cost or service objective.
In the subarea of continuous facility location, demand is not limited to a finite set of vertices:
every point on the network, i.e. on vertices and along edges, must be served.
Facilities may also be placed anywhere on edges, not only at vertices.
Distances are measured as shortest-path distance along the graph.
This model dates back to Dearing \& Francis \cite{Dearing1974}.
A typical covering version is to choose facility locations so that every point $p$ on the network is within a service distance $\delta$ of at least one facility $f_i$, i.e. $\min_i dist(p,f_i) \leq \delta$, also known as {\sc $\delta$-Covering}.
Although the feasible set is infinite, these problems often reduce to checking a finite set of critical points enabling discretization \cite{DBLP:journals/algorithmica/GrigorievHLW21,DBLP:journals/mp/HartmannLW22}.

Defensive domination models placing a limited number of defender tokens on vertices of a graph so that the network can withstand any simultaneous attack on $k$ vertices.
Formally, given an undirected graph $G$ and integers $k$ and $\ell$, we ask whether there exists a set $D \subseteq V$ with $|D| \leq \ell$ such that every attack set $A \subseteq V$ with $|A| \leq k$ can be defended as follows:
each attacked vertex $a \in A$ must be assigned a distinct token $d \in D$ in its neighborhood $N[a]$.
This is exactly the requirement that the bipartite graph between $A$ and $D$, in which an edge between $a \in A$ and $d \in D$ exists iff $a \in N[d]$, contains a matching of size $|A|$.
Thus, a feasible defense $D$ must provide enough distinct nearby defenders for every possible attacked set of size at most $k$.
For $k=1$, this reduces to the classical dominating set problem; for $k>1$, it enforces robustness against concurrent attacks and is typically computationally hard \cite{DBLP:conf/mfcs/ChaplickGK25}.
Defensive domination can be viewed as a graph-based facility-location problem where placing a defender token corresponds to opening a facility at a vertex.
An attack set $A$ of size $k$ acts like a demand scenario with $k$ clients that must be served.
A token can serve an attacked vertex only if it lies in its neighborhood.
Because each token can defend at most one attacked vertex, feasibility for a scenario is exactly the existence of a matching of size $k$ between $A$ and the chosen token set $D$.

A natural way to combine continuous facility location with defensive domination is to let both defenses and attacks be placed on points in a continuous network, i.e., at vertices and in the interior of the edges.
For this, we model the network as a metric graph $G$ whose edges have length $1$, and let $dist(\cdot,\cdot)$ denote shortest-path distance along the network. 
We place $\ell$ defender tokens at locations $D=\{d_1,\dots,d_\ell\}$, where each $d_i$ may be a vertex or a point on an edge.
An adversary then chooses an attack set $A = \{a_1, \dots, a_k\}$ of $k$ points anywhere on the network.
Defending is possible within a rational distance $\delta > 0$ if the attacked points can be matched to distinct tokens that are close enough:
there must exist an injective assignment $f: A \to D$ such that $dist(a,f(a)) \leq \delta$ for all $a \in A$, or, equivalently, we can build a bipartite graph between $A$ and $D$ as above.
We refer to this problem as {\sc Defensive $\delta$-Covering}.
This problem can model {\sc Defensive Dominating Set} by restricting the attacker and defender tokens to vertices and setting $\delta = 1$, as well as {\sc $\delta$-Covering} by setting the number of attacker tokens to $1$.

\paragraph*{Related Work}

For a background on facility location, we refer the reader to the works by Drezner \cite{drezner1996facility} and Mirchandi \& Francis \cite{mirchandani1990discrete}.
Our model was introduced by Dearing \& Francis \cite{Dearing1974}, and has since been extensively studied with multiple optimization goals.
For {\sc $\delta$-Covering}, where the goal is to cover the graph, Tamir \cite{Tamir87} and Megiddo \& Tamir \cite{MegiddoT1983} gave polynomial time algorithms for special graph classes and showed \NP-hardness for the special case $\delta = 2$.
Later, a work by Hartmann, Lendl, and Woeginger \cite{DBLP:journals/mp/HartmannLW22} showed that {\sc $\delta$-Covering} is polynomial-time solvable for unit fraction values of $\delta$ and \NP-hard otherwise, and also settled the parameterized complexity of {\sc $\delta$-Covering} for all rational values of $\delta$.
A subsequent work by Hartmann \& Janßen \cite{DBLP:conf/waoa/HartmannJ24} settled the hardness of approximation.
A dual problem to {\sc $\delta$-Covering}, where the goal is to place points whose pairwise distances are at least $\delta$, is known as {\sc $\delta$-Dispersion}.
Grigoriev, Hartmann, Lendl, and Woeginger \cite{DBLP:journals/algorithmica/GrigorievHLW21} settled the complexity of {\sc $\delta$-Dispersion} for all rational values of $\delta$, while a work by Hartmann \& Lendl \cite{DBLP:conf/mfcs/HartmannL22} considers the parameterized complexity for various graph parameters, as well as irrational values for $\delta$.
A similar problem was studied by Tamir \cite{Tamir1991}, where instead of fixing $\delta$, the number of points to be placed is fixed and the pairwise distance $\delta$ is to be optimized.
For this task, they gave a \half-approximation algorithm and showed that approximating it within any $\varepsilon > \frac{2}{3}$ is \NP-hard.
A different problem in this setting, {\sc $\delta$-Tour}, where the goal is to traverse the graph such that every point on the graph is within distance $\delta$ of the traversal, was studied by Frei, Ghazy, Hartmann, Hörsch, and Marx \cite{DBLP:conf/isaac/FreiGHHM24}.
They showed that this problem is \NP-hard for any $\delta > 0$ and analyzed the parameterized complexity and approximation properties.

For {\sc Defensive Domination}, where the size of the attack is fixed and not part of the input, Ekim, Farley, and Proskurowski  \cite{DBLP:journals/dm/EkimFP20} show \NP-completeness even for split graphs, and also show that the problem is solvable in various other special graph classes.
They also show that {\sc Defensive Domination} is {\sf coNP}-hard when the size of the attack is part of the input.
Later, Ekim, Farley, Proskurowski, and Shalom \cite{DBLP:journals/dam/EkimFPS23} also show a polynomial time result for proper interval graphs.
Henning, Pandey, and Tripathi \cite{DBLP:journals/dam/HenningPT25} show that it remains \NP-hard even on bipartite graphs, and also show {\sf APX}-hardness for bounded degree graphs.
For {\sc Defensive Domination}, where the size of the attack is part of the input, Chaplick, Gutowski, and Krawczyk \cite{DBLP:conf/mfcs/ChaplickGK25} show that the problem is $\Sigma^P_2$-complete on general graphs, and give a polynomial time algorithm for interval graphs when the defense may be a multiset.
A similar problem and its connections to cop and robber games are studied by Dereniowski, Gavenciak, and Kratochvíl \cite{DBLP:journals/tcs/DereniowskiGK19}.

\paragraph*{Our Results}

We analyze the complexity of {\sc Defensive $\delta$-Covering} in various settings.
As an intermediate step, we also show $\Sigma^P_2$-hardness for {\sc Distance-$d$ Defensive Domination} explicitly for each fixed $d$.
When $d$ is part of the input, hardness already follows from \cite{DBLP:conf/mfcs/ChaplickGK25}.
However we need hardness for each fixed $d$ as an intermediate step for our results on {\sc Defensive $\delta$-Covering} because we treat $\delta$ as fixed rather than as part of the input for that problem.
We also analyze the complexity of {\sc Distance-$d$ Defensive Domination} when the attacker may attack vertices multiple times, and show a strong relation to {\sc $k$-Tuple Dominating Set} and thus \NP-completeness.

For {\sc Defensive $\delta$-Covering} we can divide our results into two categories:
The attacker is only allowed to attack vertices, which we call {\sc Defensive $\delta$-Covering with Vertex Attack}, and the attacker is allowed to attack any point on the graph, which we simply call {\sc Defensive $\delta$-Covering}.

For the case where the attacker is only allowed to attack vertices, the complexity of the problem differs for three ranges of $\delta$:
\begin{description}
    \item[$\delta < \half$:] This case is very simple, because defender tokens can only defend at most one vertex at a time. Thus every vertex needs its own defending token(s) within distance $\delta$. Placing one token on each vertex gives an optimal defense for the set attack case, and placing $k$ tokens gives an optimal defense for the multiset attack case.
    \item[$\half \leq \delta < 1$:] If the attacker is not allowed to attack a vertex multiple times, we can show \NP-completeness by showing a strong relation to the graph factor problem. The question whether a defense countering all attacks exists becomes equivalent to the graph having a tree factor with specific properties. On the other hand, if the attacker is allowed to attack a vertex multiple times, we can construct polynomial time algorithms by reducing the problem to {\sc Capacitated $b$-Edge Cover}, which is solvable in polynomial time \cite{schrijver2003combinatorial}.
    \item[$1 \leq \delta$:] Using some tricks, we show that the problem becomes essentially equivalent to {\sc Distance-$d$ Defensive Domination}, and similarly, if the attacker is allowed to attack a multiset of vertices, the complexity drops back to \NP-hard.
\end{description}
For the case where the attacker is allowed to attack any point in the graph, we show that, surprisingly, the problem is always contained in \NP, if the size of the attack is polynomially bounded in the input size.
This stems mainly from the fact that for a non-empty interval of any length $\varepsilon>0$, we can place any number of defenders and attackers in that interval.
Our main tool for showing containment in \NP{} is a discretization argument which shows that there are always solutions where all points are placed in such a way that they can be encoded in a polynomial-length certificate.
The proof is constructive, as we show how to modify an arbitrary defensive $\delta$-cover such that it has the desired regular structure.
Our results are also summarized in \Cref{tab:my_label} with links to the respective theorems.

\begin{table}[!ht]
    \centering
    \begin{tabular}{l|c|c|c|c}
         Problem Name / Multisets & none & Defender & Attacker & both \\\hline
         Dist-$d$ Def. Dom. $(d>1)$
         & $\Sigma^P_2$ Lemma \ref{lemma:distanceDefensiveDominationCompleteness}
         & $\Sigma^P_2$ Cor. \ref{cor:distanceDefensiveDominationMultisetDefenseCompleteness}
         & \NP{} Cor. \ref{corollary:defensiveDominationMultisetBothNPComplete}
         & \NP{} Cor. \ref{corollary:defensiveDominationMultisetBothNPComplete}\\
         Def. $\delta$-C. VA $(\delta < \half)$
         & $\sf P$ Lemma \ref{lemma:defensiveDeltaCoverVertexAttackSmallerHalfInP}
         & $\sf P$ Lemma \ref{lemma:defensiveDeltaCoverVertexAttackSmallerHalfInP}
         & $\sf P$ Lemma \ref{lemma:defensiveDeltaCoverVertexAttackMultisetSmallerHalfInP}
         & $\sf P$ Lemma \ref{lemma:defensiveDeltaCoverVertexAttackMultisetSmallerHalfInP} \\
         Def. $\delta$-C. VA $(\half \leq \delta < 1)$
         & \NP{} Thm. \ref{thm:defensiveDeltaCoverVertexAttackSmallerOneCompleteness}
         & \NP{} Thm. \ref{thm:defensiveDeltaCoverVertexAttackSmallerOneCompleteness}
         & $\sf P$ Thm. \ref{thm:defensiveDeltaCoverVertexAttackMultisetInP}
         & $\sf P$ Thm. \ref{thm:defensiveDeltaCoverVertexAttackMultisetInP} \\
         Def. $\delta$-C. VA $(1 \leq \delta)$
         & $\Sigma^P_2$ Thm. \ref{thm:defensiveDeltaCoverVertexAttackGreaterOneCompleteness}
         & $\Sigma^P_2$ Thm. \ref{thm:defensiveDeltaCoverVertexAttackGreaterOneCompleteness}
         & \NP{} Thm. \ref{thm:defensiveDeltaCoverVertexAttackMultisetBothNPComplete}
         & \NP{} Thm. \ref{thm:defensiveDeltaCoverVertexAttackMultisetBothNPComplete} \\
         Def. $\delta$-C.
         & \NP{}\footnotemark{}\textsuperscript{,}\footnotemark{} Thm. \ref{thm:defensiveDeltaCoverNPCompletePolynomialK}
         & \NP{}\footnotemark[\numexpr\value{footnote}-1\relax] Thm. \ref{thm:defensiveDeltaCoverMultisetDefenseNPComplete}
         & \NP{}\footnotemark[\numexpr\value{footnote}-1\relax]\textsuperscript{,}\footnotemark[\value{footnote}] Thm. \ref{thm:defensiveDeltaCoverNPCompletePolynomialK}
         & \NP{}\footnotemark[\numexpr\value{footnote}-1\relax] Thm. \ref{thm:defensiveDeltaCoverMultisetDefenseNPComplete}
    \end{tabular}
    \caption{Problems are complete for \NP{} and $\Sigma^P_2$ and contained in {\sf P}, unless mentioned otherwise. We use VA as an abbreviation for vertex attack.}
    \label{tab:my_label}
\end{table}
\footnotetext[\numexpr\value{footnote}-1\relax]{%
    NP-hardness is not known for cases where $\delta$ is a unit fraction.}
\footnotetext[\value{footnote}]{%
    For set defenses, \NP{} membership is only shown when $k$ is polynomially bounded in the size of the input.}

\paragraph*{Organization of this Work}

The paper is organized as follows.
In \Cref{sec:preliminaries}, we introduce basic notation and existing problem definitions.
In \Cref{sec:distanceDDefensiveDomination}, we introduce {\sc Distance-$d$ Defensive Domination} as a generalization of {\sc Defensive Dominating Set} and prove its $\Sigma^P_2$-completeness as an intermediate result, which we use in the following section.
In \Cref{sec:defensiveDeltaCoverWithVertexAttack}, we investigate {\sc Defensive $\delta$-Covering with Vertex Attack} and establish either polynomial-time algorithms or \NP-/$\Sigma^P_2$-completeness results for all $\delta > 0$.
In \Cref{sec:defensiveDeltaCoverWithPointAttack}, we examine the complexity of {\sc Defensive $\delta$-Covering} and establish its $\sf NP$-completeness up to some edge cases.
\Cref{sec:conclusion} concludes the paper with a short discussion and possible future work.

\section{Preliminaries}
\label{sec:preliminaries}

For a real number $\delta$, we use $\lfloor\delta\rfloor$ to describe its integer part and $\tilde\delta := \delta - \lfloor\delta\rfloor$ its fractional part.
For a (multi-)set $X$ and element $x$ in the underlying universe, we denote the multiplicity of $x$ in $X$ by \mult{x}{X}, i.e., how many times $x$ is contained in $X$.
If $X$ is a set, then $\mult{x}{X} \in \{0, 1\}$, and if $X$ is a multiset, then $\mult{x}{X} \in \mathbb N_0$.

\paragraph*{Graph Theory}
Unless explicitly stated otherwise, we consider simple undirected graphs without isolated vertices whose edges have unit length.
A graph $G = (V,E)$ consists of a vertex set $V$ and an edge set $E$.
To make associations clearer, we denote the vertex set of a graph $G$ by $V(G)$ and the edge set by $E(G)$.

For a set $U \subseteq V(G)$ we denote by $G[U]$ the subgraph of $G$ induced by the vertices of $U$.
For a vertex $u \in V(G)$ we define its \emph{open neighborhood} as $N(u) \coloneqq \{ v \in V(G) \mid \{u, v\} \in E(G) \}$ and its \emph{closed neighborhood} as $N[u] \coloneqq N(u) \cup \{u\}$.
Similarly, for a set $U \subseteq V(G)$, we define the closed neighborhood of $U$ as $N[U] \coloneqq \bigcup_{u \in U} N[u]$ and the open neighborhood of $U$ as $N(U) \coloneqq N[U] \setminus U$.
We further generalize the closed-neighborhood concept to the distance-$d$ neighborhood of a set $U \subseteq V(G)$, where the distance-$1$ neighborhood is defined by $N_1[U] := N[U]$ and the distance-$d$ neighborhood is defined by $N_d[U] = N[N_{d-1}[U]]$.
For a vertex $v \in V(G)$, we denote by $\text{deg}(v) = |N(v)|$ the degree of $v$.

For a graph $G$ and integer $\ell \in \mathbb N$, we denote by $G_\ell$ the $\ell$-subdivision of $G$, that is, every edge in $G$ is replaced by a path with $\ell$ edges.

We use the notation of~\cite{DBLP:journals/algorithmica/GrigorievHLW21} to describe the metric space of a graph:
For an edge $\{u, v\}$ and a real $\lambda \in [0,1]$, we denote by $p(u, v, \lambda)$ the point on the edge $\{u, v\}$ that has distance $\lambda$ to $u$.
Thus $p(u, v, \lambda) = p(v, u, 1 - \lambda)$, and we will often assume that $\lambda \leq 1/2$.
Further, $p(u, v, 0) = u$ and $p(u, v, 1) = v$ and we also use $u$ to refer to the point on $u$.
We denote the continuum set of all points on the graph $G$ by $P(G)$.
With the word \emph{point} we refer to the elements of $P(G)$, while we use \emph{vertex} in the graph-theoretic sense.
The distance between two points $p, q \in P(G)$, denoted as $dist(p,q)$, is defined as the length of the shortest path from $p$ to $q$ in the underlying metric space.
The \emph{open ball} with radius $r$ around point $p$, denoted as \oball{p}{r}, is the set of points $q$ with distance $dist(q,p) < r$ from $p$.
The \emph{closed ball} with radius $r$ around point $p$, denoted as \cball{p}{r}, is the set of points $q$ with distance $dist(q,p) \leq r$ from $p$.

\paragraph*{Complexity Theory}
A language is a set $L \subseteq \{0,1\}^*$.
A language $L$ is contained in $\Sigma^P_k$ iff there exists some polynomial-time computable function $V$, which we also call the verifier, and $m_1,m_2,\ldots,m_k = |w|^{O(1)}$ such that for all $w \in \{0,1\}^*$
$$
    w \in L \iff \exists y_1 \in \{0,1\}^{m_1} \forall y_2 \in \{0,1\}^{m_2} \ldots Q y_k \in \{0,1\}^{m_k} : V(w,y_1,y_2, \ldots, y_k) = 1,
$$
where $Q = \exists$, if $k$ is odd, and $Q = \forall$, if $k$ is even.
The class of all polynomial-time computable languages $\sf P$ coincides with $\Sigma^P_0$ and the class of all non-deterministically polynomial-time computable languages $\sf NP$ coincides with $\Sigma^P_1$.
An introduction to the classes $\Sigma^P_k$ can be found in the book by Papadimitriou \cite{DBLP:books/daglib/0072413}.

A polynomial-time many-one reduction or Karp reduction from a language $L$ to a language $L'$ is a polynomial-time map $f : \{0,1\}^* \to \{0,1\}^*$ such that $w \in L$ iff $f(w) \in L'$ for all $w \in \{0,1\}^*$.
A language $L$ is $\Sigma^P_k$-hard, if every $L' \in \Sigma^P_k$ can be reduced to $L$ with a polynomial-time many-one reduction.
If $L$ is both $\Sigma^P_k$-hard and contained in $\Sigma^P_k$, it is $\Sigma^P_k$-complete.

\subsection{Problem Definitions}
This paper investigates several closely related facility-location problems on metric graphs.
We, therefore, introduce the relevant existing problems and briefly discuss their properties.
We conclude this section by stating the problem {\sc Defensive $\delta$-Covering} and its variants, which we will examine in this paper.

We start with a variant of continuous facility location, which is known as {\sc $\delta$-Covering}.
This problem serves as one foundation for this paper.
This setting has been used to generalize the classic problems {\sc Vertex Cover} and {\sc Dominating Set} into a continuous setting \cite{DBLP:journals/mp/HartmannLW22,DBLP:conf/waoa/HartmannJ24,MegiddoT1983,Tamir1991,Tamir87}, which leads to the following problem.
For a set $S \subseteq P(G)$, we say that $S$ is a $\delta$-cover of $G$, if for every point $p \in P(G)$ there exists a point $q \in S$ such that $dist(p,q) \leq \delta$.

\begin{definition}[$\delta$-Covering \cite{DBLP:journals/mp/HartmannLW22}]
    \label{def:deltaCovering}
    Given a graph $G$ and $k \in \mathbb N$, decide if there exists $S \subseteq P(G)$ with $|S| \leq k$ such that $S$ is a $\delta$-cover of $G$.
\end{definition}

The other foundational problem for our paper is defensive domination or defensive dominating set.
Defensive domination is a two-player game between a defender and an attacker.
The game board is a graph.
The game consists of a move by the defender followed by a move by the attacker.
First, the defender places $\ell$ tokens on the vertices.
Second, the attacker places $k$ tokens on the vertices.
The defender wins if it is possible to match each attacker token to a distinct defender token, while each defender token $d$ on vertex $v_d$ can only be matched to attacker tokens $a$ such that $v_a$ is in the neighborhood $N[v_d]$ of $v_d$.
Otherwise, the attacker wins.

Based on this game, we make the following definitions.
A multiset of at most $k$ attacker tokens $A$ is called a $k$-attack.
A multiset of at most $\ell$ defender tokens $D$ is called an $\ell$-defense.
We say that an $\ell$-defense \emph{counters} a $k$-attack, if it is possible to match a distinct defender token $d \in D$ to each attacker token $a \in A$ such that $v_a\in N[v_d]$.

The original variant of defensive domination considers $A$ and $D$ to be subsets of the vertices of the graph, i.e. defender and attacker are able to place at most one token on every vertex of the graph.
Further, a defender token $d$ is able to reach the closed neighborhood $N[v_d]$.
Accordingly, {\sc Defensive Dominating Set} is defined as follows.

\begin{definition}[{\sc Defensive Dominating Set}]
    \label{def:defensiveDominatingSet}
    Given a graph $G$ and $k, \ell \in \mathbb N$, decide whether $G$ admits an $\ell$-defense $D \subseteq V(G)$ that counters every possible $k$-attack $A \subseteq V(G)$ in $G$.
\end{definition}

We generalize the Defensive Dominating Set Problem to the continuous setting, where the attacks and defenses may be points in $P(G)$ instead of just vertices.
A defense $D$ then counters an attack $A$, if for every point $p_a \in A$ we can match a distinct point $p_d \in D$ to $p_a$ such that $dist(p_a, p_d) \leq \delta$.
We also use the defender token synonymously with the point (resp. vertex) it is placed on.

\begin{definition}[{\sc Defensive $\delta$-Covering With Point Attack}]
    \label{def:defensiveDeltaCover}
    Given a graph $G$ and $k, \ell \in \mathbb N$, decide whether $G$ admits an $\ell$-defense $D \subseteq P(G)$ that counters every possible $k$-attack $A \subseteq P(G)$ in $G$.
\end{definition}

Since this is the most general problem variant, we simply write {\sc Defensive $\delta$-Covering} instead of {\sc Defensive $\delta$-Covering With Point Attack}.
An important observation for defensive dominating sets and defensive $\delta$-covers is that the matching between attack tokens and defender tokens can be computed as a matching over a bipartite graph.
The vertices on the one side are the attacker tokens and on the other side are the defender tokens.
An edge between an attacker token $a \in A$ and defender token $d \in D$ exists if and only if $d$ reaches $a$, i.e., $a \in B^\leq(d,\delta)$.
Hence, we can state the following variant of Hall's marriage theorem, which was already observed by the authors of \cite{DBLP:conf/mfcs/ChaplickGK25} and \cite{DBLP:journals/dm/EkimFP20}. 
We define $count_D(A) = |D \cap B^\leq(A,\delta)|$ (counting each defender token copy) for the given $\delta$.
\begin{note}
    A defense $D$ counters every $k$-attack in $G$ if and only if for every $k$-attack $A$, we have $|A| \leq count_D(A)$.
\end{note}

We start by dividing this problem into different variants which exhibit different behaviors.
The first variant is {\sc Defensive $\delta$-Covering With Vertex Attack}, in which the defender is able to place its tokens onto every point $P(G)$, the metric space associated with $G$, while the attacker is constrained to only attack vertices, i.e., $A \subseteq V(G)$.

\begin{definition}[{\sc Defensive $\delta$-Covering With Vertex Attack}]
    \label{def:defensiveDeltaCoverWithVertexAttack}
    Given a graph $G$ and $k, \ell \in \mathbb N$, decide whether $G$ admits an $\ell$-defense $D \subseteq P(G)$ that counters every possible $k$-attack $A \subseteq V(G)$ in $G$.
\end{definition}

It is further possible to distinguish between set and multiset defenses or attacks.
In \cite{DBLP:conf/mfcs/ChaplickGK25} the authors also consider the corresponding relaxation of Defensive Dominating Set, where $D$ is allowed to be a multiset, i.e., multiple defenders may be placed on the same vertex.
We refer to that problem as {\sc Defensive Dominating Set with Multiset Defense}. 
The authors of \cite{DBLP:conf/mfcs/ChaplickGK25} indicate that it is $\Sigma^P_2$-complete as well.

Further, we consider the relaxation where $A$ is allowed to be a multiset, i.e., a vertex may be attacked multiple times, and then also has to be defended multiple times.
We refer to that problem as {\sc Defensive Dominating Set with Multiset Attack}. 
The authors of \cite{DBLP:conf/mfcs/ChaplickGK25} mention that this is a different problem, which we prove in the following.

\section{Distance-$d$ Defensive Domination}
\label{sec:distanceDDefensiveDomination}

We start the analysis of the problem family of {\sc Defensive $\delta$-Covering} with an intermediary result, which concentrates on {\sc Defensive Dominating Set}, in which a defender token on a vertex $v$ is not only able to defend the distance-$1$ neighborhood $N[v]$ of a vertex $v$, but the distance-$d$ neighborhood $N_d[v]$.
Accordingly, we define {\sc Distance-$d$ Defensive Dominating Set} as follows.

\begin{definition}[{\sc Distance-$d$ Defensive Dominating Set}]
    \label{def:distanceDefensiveDominatingSet}
    Given a graph $G$ and $k, \ell \in \N$, decide whether $G$ admits an $\ell$-defense $D \subseteq V(G)$ that counters every possible $k$-attack $A \subseteq V(G)$ in $G$.
    A defense $D \subseteq V(G)$ counters an attack $A \subseteq V(G)$, if for every $v_a \in A$ there is a distinct $v_d \in D$ such that $dist(v_a, v_d) \leq d$.
\end{definition}

It is easy to see that {\sc Distance-$d$ Defensive Dominating Set} is $\Sigma^P_2$-complete for $d=1$, as then it is equivalent to regular Defensive Dominating Set.
We further prove $\Sigma^P_2$-completeness for all $d \in \N$.
For this, we use the reduction from \textsc{Minimum Cardinality Clique Interdiction} to \textsc{Defensive Dominating Set} of \cite{DBLP:conf/mfcs/ChaplickGK25}, which generates instances that have a useful property encapsulated in the following \Cref{lem:homevertex}.
Based on \Cref{lem:homevertex}, we can enhance the construction with a subdivision gadget to obtain $\Sigma^P_2$-hardness of {\sc Distance-$d$ Defensive Dominating Set}.

The problem {\sc Minimum Cardinality Clique Interdiction} is $\Sigma^P_2$-complete \cite{DBLP:journals/amai/Rutenburg93,DBLP:conf/mfcs/ChaplickGK25,DBLP:journals/tcs/GruneW26} and defined as follows.
The input is an undirected graph $G$ together with integers $s$ and $t$.
The question is whether there is a set $S \subseteq V(G)$ with $|S| \leq s$ such that there is no clique $T \subseteq V(G) \setminus S$ in $G$ of size at least $t$.

For a defender token $\tau \in D$, we define a \emph{home vertex} $h(\tau)$, if the home vertex is a vertex such that $\tau$ is able to defend at most the vertices in $N[h(\tau)]$.

\begin{lemma}\label{lem:homevertex}
    Let $(\hat G,s,t)$ be an instance of {\sc Minimum Cardinality Clique Interdiction} and $(G,k,\ell)$ be the \textsc{Defensive Dominating Set} instance constructed in the reduction of \cite{DBLP:conf/mfcs/ChaplickGK25}.
    Let $D$ be a collection of at most $\ell$ tokens in which every token $\tau$ has a home vertex $h(\tau) \in V(G)$.
    If $D$ counters every $k$-attack $A \subseteq V(G)$, then $(\hat G,s,t)$ is a yes-instance.
\end{lemma}
\begin{proof}
    The instance $(G,k,\ell)$ is identical to the one in Theorem $3$ of \cite{DBLP:conf/mfcs/ChaplickGK25}.
    Therefore, $k = |V(\hat G)| + s$ and $\ell = 4(|V(\hat G)|+s) + t\,|V(\hat G)| - (t+1)$.\footnote{In Theorem $3$ of \cite{DBLP:conf/mfcs/ChaplickGK25}, $\ell = 4(|V(\hat G)|+s) + t\,|V(\hat G)| - (t-1)$, however in the proof they use $\ell$ as if $\ell = 4(|V(\hat G)|+s) + t\,|V(\hat G)| - (t+1)$.}
    We follow the correctness proof for the backward direction, and use their notation.

    For $Y \subseteq V(G)$, let $T_Y$ denote the tokens with home in $Y$ (distinct tokens may share the same home).
    Every token that is able to defend a set $A$ has its home in $N[A]$.
    Therefore we obtain Hall's inequality for a successful defense of tokens $T$:
    \begin{equation} \label{eq:homevertex:hall}
        |A| \leq |T_{N[A]}| \text{ for every $k$-attack } A.
    \end{equation}
    Now consider the observations of the backward direction of the correctness proof.
    Observations 1-4 and 6 use~\eqref{eq:homevertex:hall}.
    They remain valid if the number of defenders in the given set $Y$ is $|T_Y|$ as they follow from applying~\eqref{eq:homevertex:hall} to a suitable $k$-attack.
    Since $|I_3| = |V(\hat G)|+s+\ell$, there are at most $\ell$ homes in $I_3$ as in Observation~3.
    Thus, an attack $A$ on $|V(\hat G)|+s$ vertices of $I_3$ that are not homes yields $N[A]=A \cup V$, and we obtain $|T_V| \geq |V(\hat G)|+s$.
    Overall, we obtain $|T_{I_1 \cup Q_1}|$, $|T_{I_4 \cup Q_4}|$,$|T_V| \geq |V(\hat G)|+s$.
    
    For every $v\in V(\hat G)$, we obtain $|T_{I_v\cup I'_v}| \geq t$ and $|T_{I_v\cup I'_v}|\ge t+1$ if $T_{\{v',v''\}}=\emptyset$.
    For every $v$ with $T_{\{v',v''\}}=\emptyset$, we change the home of one token of $T_{I_v\cup I'_v}$ to $v'$.
    Thus, we have $|T_{I_v\cup I'_v}| \geq t$ and $n_v := |T_{\{v',v''\}}| \geq 1$ for all $v \in V(\hat G)$.
    Note that no token with home in $I_v\cup I'_v$ is able to defend a vertex in $I_2 \cup E$.
    Since $|T_V| = \sum_{v \in V(\hat G)} n_v \geq |V(\hat G)|+s$ and by counting all assigned tokens, we obtain
    \begin{equation} \label{eq:homevertex:remainingbudget}
        |T_{I_2 \cup Q_2 \cup E}| + |T_V| \leq 2(|V(\hat G)|+s)-(t+1).
    \end{equation}

    Now let $R = |T_V| - |V(\hat G)| \geq s$ be the number of surplus tokens in $V$.
    We choose the set $X \subseteq V(\hat G)$ with $|X| = s$ such that $X$ contains as many vertices with $n_v \geq 2$ as possible, and define a set $T' \subseteq V(G)$ as follows.
    \begin{itemize}
        \item For every $v \in V(\hat G)$, one token of $T_{\{v',v''\}}$ is placed on $v'$.
        \item For every $v \in X$ with $n_v \geq 2$, we place one additional tokens on $v''$.
        \item For every $w\in X$ with $n_w=1$, a token from the surplus of some $v\in X$ with $n_v\ge 3$ is placed on $w''$.
    \end{itemize}
    This is possible since, if such a $w$ exists, every vertex with $n_v \geq 2$ lies in $X$, so all $R \geq s$ surplus tokens belong to vertices of $X$.
    All remaining tokens of $T_{V\cup I_2\cup Q_2\cup E}$ are placed on distinct vertices of $Q_2$, which is possible since, by~\eqref{eq:homevertex:remainingbudget}, there are $|T_{I_2\cup Q_2\cup E}| + R - s \leq |V(\hat G)|+s-(t+1) = |Q_2|$ of them.

    Now suppose for a contradiction that $\hat G$ contains a clique of size at least $t$ that avoids $X$, and let $C$ be $t$ of its vertices.
    Let $A_C$ consist of $I_2$ and the $\binom{t}{2}$ vertices $e'\in E$ for the edges $e$ of $\hat G[C]$.
    By~\eqref{eq:homevertex:hall}, at least $|V(\hat G)|+s$ tokens reach $A_C$, and their homes lie in $N[A_C]=I_2 \cup Q_2 \cup E \cup \{v',v'':v\in C\}$.
    The tokens re-homed above do not reach $A_C$, and neither do tokens with home in $\{w',w''\}$ for $w \in X$, as $w \not\in C$.
    Every other token reaching $A_C$ is placed in $T'$ on a vertex of $Q_2$, which is adjacent to all of $I_2 \cup E$, or on $v'$ for some $v \in C$, where $N[v'] = N[v'']$.
    Hence at least $|V(\hat G)| + s$ vertices of $T'$ reach $A_C$.
    However, all of them lie in $(T' \cap Q_2) \cup \{v' : v\in C\}$, which has at most $|V(\hat G)|+s-(t+1)+t < |V(\hat G)|+s$ vertices, a contradiction.
    Since $|X|=s$, the proof is complete.
\end{proof}

\begin{lemma}
    \label{lemma:distanceDefensiveDominationCompleteness}
    {\sc Distance-$d$ Defensive Dominating Set} is $\Sigma^P_2$-complete for any $d \in \N$.
\end{lemma}
\begin{proof}

    We prove the containment in $\Sigma^P_2$ by providing an $\exists\forall$-witness of polynomial length.
    With the $\exists$-quantified string, we encode the defense and with the $\forall$-quantified string, we encode the attack, each as a subset of $V(G)$.
    Then, the verifier computes a matching in the bipartite graph consisting of the attacked vertices on one side and the defender tokens on the other side; an edge exists between the attacked vertex and the defender token if and only if the defender token is at distance at most $d$ from the attacked vertex.
    If and only if the matching covers all attacked vertices for any attack $A$, the given instance is a yes-instance.

    To show hardness, we reduce from {\sc Minimum Cardinality Clique Interdiction} via the instance of \cite{DBLP:conf/mfcs/ChaplickGK25}.
    That is, let $(G,k,\ell)$ be the instance of {\sc Defensive Domination} as in \Cref{lem:homevertex}, where $G$ is the graph, $k$ is the number of attackers, and $\ell$ the number of defenders.
    We construct an instance $(G', k', \ell')$ for {\sc Distance-$d$ Defensive Domination} such that $(G,k,\ell)$ is a yes-instance if and only if $(G', k', \ell')$ is a yes-instance.

    For $d=1$, the claim is \cite[Theorem~3]{DBLP:conf/mfcs/ChaplickGK25}.
    In the following let $d \geq 2$.
    For the construction, we substitute each edge $\{u,v\} \in E(G)$ by the gadget depicted in \Cref{fig:distanceDefensiveDomination:odd}, if $d$ is odd, and by the gadget depicted in \Cref{fig:distanceDefensiveDomination:even}, if $d$ is even.
    We obtain the gadget by subdividing the edge $\{u,v\}$ $d-1$ times and appending a path of $\lfloor \frac{d+1}{2} \rfloor$ vertices in the middle of the subdivided path between $u$ and $v$.
    If $d$ is even, this path is connected to the middle vertex.
    If $d$ is odd, it is connected to the two middle vertices and the path is appended to both of them.
    The appended path is followed by a clique $C^{\{u,v\}}_k$ of size $k$, another path of $d-1$ vertices, and an independent set $I^{\{u,v\}}_k$ of size $k$.
    Additionally, we set $k' = k$ and $\ell' = k|E(G)| + \ell$.
    
    We prove that the additional $k|E(G)|$ defender tokens need to be placed into the cliques $C^{\{u,v\}}_k$ for all edges $\{u,v\} \in E(G)$, while the remaining $\ell$ defenders counter every $k$-attack in the original instance $(G,k,\ell)$ of {\sc Defensive Domination}.

    For a contradiction, let us assume that the defender places fewer than $k$ tokens into the gadget (excluding $u$ and $v$).
    Then, the attacker attacks all vertices in $I^{\{u,v\}}_k$ and wins.
    In order to defend $I^{\{u,v\}}_k$, the defender must place $k$ tokens into $C^{\{u,v\}}_k \cup I^{\{u,v\}}_k$, and the path between $C^{\{u,v\}}_k$ and $I^{\{u,v\}}_k$.
    All other vertices have distance more than $d$ from $I^{\{u,v\}}_k$.
    Thus, none of the defender tokens of the gadget is able to defend $u$ or $v$.    
    Furthermore, to defend all the vertices between $u$, $v$, and $C^{\{u,v\}}_k$, all defender tokens can be placed into $C^{\{u,v\}}_k$.
    Thus, by placing $k$ of its tokens into $C^{\{u,v\}}_k$, the defender is able to defend all possible attacks that include any vertices of the gadget (excluding $u$ and $v$).
    In conclusion, by the above strategy, all original vertices $V(G)$ are not defended, while all newly introduced vertices are defended.

    For the forward direction, consider the original $\ell$-defense $D$ in graph $G$.
    Using the $(k|E(G)|+\ell)$-defense $D' = D \cup \bigcup_{\{u,v\} \in E(G)} C^{\{u,v\}}_k$, all vertices in the edge gadgets are defended by the corresponding $C^{\{u,v\}}_k$.
    Furthermore, since each edge in $G$ is subdivided into a path of length $d$ in $G'$, the distance between vertices $u, v$ is $d$ if $\{u,v\} \in E(G)$.
    Thus, each distance-$d$ $\ell$-defense to every $k$-attack in $G$ defends all original vertices $V(G)$ in $G'$.
    In conclusion, each $\ell$-defense $D$ that counters every $k$-attack in $G$ can be extended to an $(k|E(G)|+\ell)$-defense $D'$ that counters every $k$-attack in $G'$.
    
    For the backward direction, consider a $(k|E(G)|+\ell)$-defense $D'$ that counters every $k$-attack in $G'$. 
    By the above argument $k|E(G)|$ defender tokens need to be placed into $C^{\{u,v\}}_k$ for each $\{u,v\} \in E(G)$.
    This leaves the original vertices $v \in V(G)$ undefended.
    Therefore, the remaining $\ell$ defender tokens need to defend at least the vertices $V(G)$ in $G'$.
    A defender token on $w \in V(G)$ is able to defend exactly the vertices of $N[w]$ of $G$ because the distance between $(w,x) \in G$ with $\{w,x\} \in E(G)$ is $d$ in $G'$ and the distance between $(w,x) \in G$ with $\{w,x\} \not\in E(G)$ is at least $2d$ in $G'$.
    Any defender token that is placed on vertices within a gadget of some edge $\{u,v\}$ can defend at most $u$ and $v$ but no $w \in V(G) \setminus \{u,v\}$ in $G'$.
    In this case, we set its home to $u$.
    Since \Cref{lem:homevertex} applies, the proof is complete.
    
\end{proof}

\begin{figure}[t]
    \begin{subfigure}{0.48\textwidth}
        \centering
        \scalebox{0.36}{
            \begin{tikzpicture}[
                cross/.style={%
                    path picture={%
                        \draw[thick] (path picture bounding box.south west) -- (path picture bounding box.north east);
                        \draw[thick] (path picture bounding box.north west) -- (path picture bounding box.south east);
                    }
                }
            ]            
                \node[circle,draw,thick,minimum size=0.75cm] (u) at (-1,8) {}; \node[] (ul) at (-1.75,8) {\Huge $u$};
                \node[circle,draw,thick,minimum size=0.75cm] (v) at (-1,4) {}; \node[] (vl) at (-1.75,4) {\Huge $v$};
                \node[circle,draw] (1) at (-1,6.75) {};
                \node[circle,draw] (3) at (-1,5.25) {};
        
                \draw[black] (u) -- (1);
                \draw[black] (v) -- (3);
                \draw[black] (1) -- (3);
        
                \node[circle,draw] (21) at (0,6) {};
                \node[circle,draw] (22) at (1,6) {};
        
                \draw[black] (1) -- (21);
                \draw[black] (3) -- (21);
                \draw[black] (21) -- (22);
        
                \node[circle,draw,thick,minimum size=3cm] (C) at (4,6) {}; \node[] (Cl) at (5,8) {\Huge $C_k$};
                \node[circle,draw,thick,minimum size=0.75cm,cross] (v) at (4,6) {};
                \node[circle,draw,thick,minimum size=0.75cm,cross] (v) at (4,7) {};
                \node[circle,draw,thick,minimum size=0.75cm,cross] (v) at (3.45,5.2) {};
                \node[circle,draw,thick,minimum size=0.75cm,cross] (v) at (4.55,5.2) {};
                \node[circle,draw,thick,minimum size=0.75cm,cross] (v) at (4.9,6.25) {};
                \node[circle,draw,thick,minimum size=0.75cm,cross] (v) at (3.1,6.25) {};
                \node[] (Cl1) at (2.85,5) {};
                \node[] (Cl2) at (2.5,6) {};
                \node[] (Cl3) at (2.85,7) {};
                \node[] (Cr1) at (5.15,5) {};
                \node[] (Cr2) at (5.5,6) {};
                \node[] (Cr3) at (5.15,7) {};         
                    
                \draw[draw=black] (10,3.5) rectangle ++(1,5); \node[] (Il) at (11.5,8) {\Huge $I_k$};
                \node[circle,draw,thick,minimum size=0.75cm] (v) at (10.5,4) {};
                \node[circle,draw,thick,minimum size=0.75cm] (v) at (10.5,5) {};
                \node[] (v) at (10.5,6) {\vdots};
                \node[circle,draw,thick,minimum size=0.75cm] (v) at (10.5,7) {};
                \node[circle,draw,thick,minimum size=0.75cm] (v) at (10.5,8) {};
                \node[] (Il1) at (10,4.75) {};
                \node[] (Il2) at (10,6) {};
                \node[] (Il3) at (10,7.25) {};
                    
                \draw[lightgray,thick] (Cl1) -- (22);
                \draw[lightgray,thick] (Cl2) -- (22);
                \draw[lightgray,thick] (Cl3) -- (22);
    
                \node[circle,draw] (C21) at (7,6) {};
                \node[circle,draw] (C23) at (8,6) {};
    
                \draw[black] (C21) -- (C23);
    
                \draw[lightgray,thick] (Cr1) -- (C21);
                \draw[lightgray,thick] (Cr2) -- (C21);
                \draw[lightgray,thick] (Cr3) -- (C21);
                    
                \draw[lightgray,thick] (C23) -- (Il1);
                \draw[lightgray,thick] (C23) -- (Il2);
                \draw[lightgray,thick] (C23) -- (Il3);
    
                \draw[decoration={brace,raise=70pt, amplitude=10pt},decorate] ([xshift=-2pt]21.west) -- node[above=80pt] {\huge $\frac{d
                +1}{2}$} ([xshift=2pt]22.east);
    
                \draw[decoration={brace, amplitude=10pt,raise=20pt},decorate] (-1,5) -- node[left=35pt] {\huge $d-1$} (-1,7);
    
                \draw[decoration={brace,raise=70pt, amplitude=10pt},decorate] ([xshift=-2pt]C21.west) -- node[above=80pt] {\huge $d-1$} ([xshift=2pt]C23.east);
            \end{tikzpicture}
        }
        \caption{Gadget for odd distances (here: $d=3$).}
        \label{fig:distanceDefensiveDomination:odd}
    \end{subfigure}
    \hfill
    \begin{subfigure}{0.48\textwidth}
        \centering
        \scalebox{0.36}{
            \begin{tikzpicture}[
                cross/.style={%
                    path picture={%
                        \draw[thick] (path picture bounding box.south west) -- (path picture bounding box.north east);
                        \draw[thick] (path picture bounding box.north west) -- (path picture bounding box.south east);
                    }
                }
            ]            
                \node[circle,draw,thick,minimum size=0.75cm] (u) at (-1,8) {}; \node[] (ul) at (-1.75,8) {\Huge $u$};
                \node[circle,draw,thick,minimum size=0.75cm] (v) at (-1,4) {}; \node[] (vl) at (-1.75,4) {\Huge $v$};
                \node[circle,draw] (1) at (-1,7) {};
                \node[circle,draw] (2) at (-1,6) {};
                \node[circle,draw] (3) at (-1,5) {};
        
                \draw[black] (u) -- (1);
                \draw[black] (v) -- (3);
                \draw[black] (1) -- (2);
                \draw[black] (2) -- (3);
        
                \node[circle,draw] (21) at (0,6) {};
                \node[circle,draw] (22) at (1,6) {};
        
                \draw[black] (2) -- (21);
                \draw[black] (21) -- (22);
        
                \node[circle,draw,thick,minimum size=3cm] (C) at (4,6) {}; \node[] (Cl) at (5,8) {\Huge $C_k$};
                \node[circle,draw,thick,minimum size=0.75cm,cross] (v) at (4,6) {};
                \node[circle,draw,thick,minimum size=0.75cm,cross] (v) at (4,7) {};
                \node[circle,draw,thick,minimum size=0.75cm,cross] (v) at (3.45,5.2) {};
                \node[circle,draw,thick,minimum size=0.75cm,cross] (v) at (4.55,5.2) {};
                \node[circle,draw,thick,minimum size=0.75cm,cross] (v) at (4.9,6.25) {};
                \node[circle,draw,thick,minimum size=0.75cm,cross] (v) at (3.1,6.25) {};
                \node[] (Cl1) at (2.85,5) {};
                \node[] (Cl2) at (2.5,6) {};
                \node[] (Cl3) at (2.85,7) {};
                \node[] (Cr1) at (5.15,5) {};
                \node[] (Cr2) at (5.5,6) {};
                \node[] (Cr3) at (5.15,7) {};         
                    
                \draw[draw=black] (11,3.5) rectangle ++(1,5); \node[] (Il) at (12.5,8) {\Huge $I_k$};
                \node[circle,draw,thick,minimum size=0.75cm] (v) at (11.5,4) {};
                \node[circle,draw,thick,minimum size=0.75cm] (v) at (11.5,5) {};
                \node[] (v) at (11.5,6) {\vdots};
                \node[circle,draw,thick,minimum size=0.75cm] (v) at (11.5,7) {};
                \node[circle,draw,thick,minimum size=0.75cm] (v) at (11.5,8) {};
                \node[] (Il1) at (11,4.75) {};
                \node[] (Il2) at (11,6) {};
                \node[] (Il3) at (11,7.25) {};
                    
                \draw[lightgray,thick] (Cl1) -- (22);
                \draw[lightgray,thick] (Cl2) -- (22);
                \draw[lightgray,thick] (Cl3) -- (22);
    
                \node[circle,draw] (C21) at (7,6) {};
                \node[circle,draw] (C22) at (8,6) {};
                \node[circle,draw] (C23) at (9,6) {};
    
                \draw[black] (C21) -- (C22);
                \draw[black] (C22) -- (C23);
    
                \draw[lightgray,thick] (Cr1) -- (C21);
                \draw[lightgray,thick] (Cr2) -- (C21);
                \draw[lightgray,thick] (Cr3) -- (C21);
                    
                \draw[lightgray,thick] (C23) -- (Il1);
                \draw[lightgray,thick] (C23) -- (Il2);
                \draw[lightgray,thick] (C23) -- (Il3);
    
                \draw[decoration={brace,raise=70pt, amplitude=10pt},decorate] ([xshift=-2pt]21.west) -- node[above=80pt] {\huge $\frac{d}{2}$} ([xshift=2pt]22.east);
    
                \draw[decoration={brace, amplitude=10pt,raise=20pt},decorate] (-1,4.75) -- node[left=35pt] {\huge $d-1$} (-1,7.25);
    
                \draw[decoration={brace,raise=70pt, amplitude=10pt},decorate] ([xshift=-2pt]C21.west) -- node[above=80pt] {\huge $d-1$} ([xshift=2pt]C23.east);
            \end{tikzpicture}
        }
        \caption{Gadget for even distances (here: $d=4$).}
        \label{fig:distanceDefensiveDomination:even}
    \end{subfigure}
    \caption{Gadgets replacing the edges for the construction of \Cref{lemma:distanceDefensiveDominationCompleteness}.}
\end{figure}

\begin{corollary}
    \label{cor:distanceDefensiveDominationMultisetDefenseCompleteness}
    {\sc Distance-$d$ Defensive Domination with Multiset Defense} is $\Sigma^P_2$-complete for any $d \in \N$.
\end{corollary}
\begin{proof}

    For the containment, note that $|V|$ defender tokens are enough to defend all possible attacks since the attacker can attack each vertex only once.
    Thus, the number of defender tokens is polynomial and we can reuse the $\exists\forall$-witness and verifier from above.

    The construction and proof of \Cref{lemma:distanceDefensiveDominationCompleteness} can be reused, since the forward direction is proven using a set and the backward direction utilizing \Cref{lem:homevertex} does not require the tokens to have distinct home vertices.
    
\end{proof}

\subsection{Multiset Attack}

A variant of {\sc Distance-$d$ Dominating Set} emerges, if we consider multiset attacks.
As we will prove, this variant is equivalent to the problem {\sc Distance-$d$ $k$-tuple Dominating Set}, which is defined as follows.

\begin{definition}[{\sc Distance-$d$ $k$-tuple Dominating Set}]
    \label{def:kTupleDominatingSet}
    Given a graph $G$ and $k, \ell \in \mathbb N$, decide whether there is a (multi-)set $D \subseteq V(G)$ of size at most $\ell$ such that for every vertex $v \in V(G)$ it holds that $\sum_{u \in N_d[v]} \mult{u}{D} \geq k$.
\end{definition}

\begin{lemma}
    \label{lemma:kTupleDominatingSetEquivalence}
    {\sc Distance-$d$ Defensive Dominating Set with Multiset Attack} is equivalent to {\sc Distance-$d$ $k$-tuple Dominating Set}, regardless of whether the defense may be a multiset or not.
\end{lemma}
\begin{proof}
    Let $D$ be an $\ell$-defense that counters any $k$-attack $A$ in $G$, where $A$ may be a multiset.
    If $D$ is not a {\sc $k$-tuple Dominating Set}, there is a vertex $v \in V(G)$ such that $\sum_{u \in N_d[v]} \mult{u}{D} \leq k-1$.
    Then $D$ does not defend the attack placing $k$ tokens on $v$, contradiction.

    On the other hand, let $D$ be a {\sc $k$-tuple Dominating Set} of size $\ell$.
    For an arbitrary $k$-attack $A$, we build the following bipartite graph $G' = (A \cup B, E')$.
    One partition is $A$ with one vertex per attacker token copy, and for each $v_a \in A$ we add the vertices of $D$ to $B$ that dominate $v_a$ and connect them to $v_a$.
    If such a dominating vertex is contained multiple times in $D$, we add the respective number of copies to $B$.
    Since $D$ is a $k$-tuple Dominating Set, each vertex $v_a$ has degree at least $k$ in $G'$.
    Further, $|A| = k$ and thus for every set $A' \subseteq A$ it holds that $|N(A')| \geq |A'|$ in $G'$.
    Therefore, by Hall's Theorem, there is a matching of size $k$ in $G'$, and thus $D$ counters $A$.
    Thus $D$ is also an $\ell$-defense that counters every $k$-attack.
    
    The arguments above work provided that either both the $\ell$-defense and the {\sc $k$-tuple Dominating Set} may be a multiset or neither may be a multiset.
\end{proof}

It is easy to prove that {\sc Distance-$d$ $k$-tuple Dominating Set} is \NP-complete by a reduction from {\sc Distance-$d$ Dominating Set}.

\begin{lemma}
    \label{lemma:kTupleDominatingSetNPComplete}
    {\sc Distance-$d$ $k$-tuple Dominating Set} is \NP-complete. 
\end{lemma}
\begin{proof}
   The \NP-hardness easily follows by reduction from {\sc Distance-$d$ Dominating Set} as {\sc Distance-$d$ $k$-tuple Dominating Set} is equivalent to {\sc Distance-$d$ Dominating Set} for $k = 1$.
   For containment in \NP, we can simply encode the elements of $D$ as well as their multiplicity.
   It is easy to see that the multiplicity of each element is at most $k$.
   Therefore the length of the certificate is polynomially bounded.
   To verify the certificate, we check the size of $D$, and compute $\sum_{u \in N_d[v]} \mult{u}{D}$ for every $v \in V(G)$ and accept if every result is at least $k$.
   Both can be done in polynomial time as the encoding length of the occurring numbers is at most $\mathcal O(log(|V(G)|k+1))$, which completes the proof.
\end{proof}

From the previous two lemmas, we can infer the following.

\begin{corollary}
    \label{corollary:defensiveDominationMultisetBothNPComplete}
    {\sc Distance-$d$ Defensive Dominating Set with Multiset Attack} is \NP-complete, regardless of whether the defense may be a multiset or not.
\end{corollary}

\section{Defensive $\delta$-Covering with Vertex Attack}
\label{sec:defensiveDeltaCoverWithVertexAttack}

In this section, we study the problem {\sc Defensive $\delta$-Covering With Vertex Attack}, in which the attacker is only allowed to attack elements of $V(G)$.
We start this section by establishing structural observations that help us to discretize the possible placements of defender tokens onto points in $P(G)$.
The defenses admit different structures for different $\delta$, which we can group into three groups, where $\tilde\delta = 0$ or $\tilde\delta = \frac{1}{2}$, $\tilde\delta \in (0,\frac{1}{2})$, and $\tilde\delta \in (\frac{1}{2},1)$.

The following lemma shows that for $\delta$ with a fractional part $\tilde\delta < \half$, an optimal defense may place all tokens on vertices, or very close to vertices if the defense should not be a multiset.

\begin{lemma}
    \label{lemma:defensiveDeltaCoverVertexAttackFractionalPartSmallerHalf}
    For any graph $G$ and any $\delta$ with $\tilde\delta < \half$, it holds that 
    \begin{enumerate}
        \item for any point $p = p(u,v,\lambda) \in P(G)$ with $\lambda \leq \half$, we have
        $$\cball{p}{\delta} \cap V(G) \subseteq \cball{p(u,v,0)}{\delta} \cap V(G).$$
        \item for any two points $p = p(u,v,0), q = p(u,v,\lambda) \in P(G)$ with $\lambda \leq \tilde\delta$, we have $$\cball{p}{\delta} \cap V(G) = \cball{q}{\delta} \cap V(G).$$
    \end{enumerate}
\end{lemma}
\begin{proof}
    We start with the proof of the first statement.
    Consider an arbitrary vertex $v^*$ that is $\delta$-covered by $p$.
    Then, $v^*$ has distance at most $\lfloor\delta - \lambda\rfloor \leq \lfloor\delta\rfloor$ to $u$ or at most $\lfloor\delta - 1 + \lambda\rfloor \leq \lfloor\delta\rfloor - 1$ to $v$, since $\tilde\delta + \lambda < 1$.
    Thus $u$ has distance at most $\lfloor\delta\rfloor$ to $v^*$ and thus $u$ $\delta$-covers $v^*$.

    For the second statement, we consider an arbitrary vertex $v^*$ that is $\delta$-covered by $q$.
    Then, $v^*$ has distance at most $\lfloor\delta - \lambda\rfloor$ to $u$ or at most $\lfloor\delta - 1 + \lambda\rfloor$ to $v$.
    Furthermore, since $\lambda \leq \tilde \delta < \frac{1}{2}$, $\lfloor\delta - 1 + \lambda\rfloor = \lfloor\delta\rfloor - 1$.
    That is, $v^*$ must have distance at most $\lfloor\delta-1\rfloor$ to $v$ or distance at most $\lfloor\delta\rfloor$ to $u$, and thus it has distance at most $\lfloor\delta\rfloor$ to $u$ in both cases and any vertex $\delta$-covered by $q$ is also $\delta$-covered by $p$.
    Now consider an arbitrary vertex $v^{**}$ that is $\delta$-covered by $p$.
    Then $v^{**}$ has distance at most $\lfloor\delta\rfloor$ to $u$ and thus also distance at most $\delta$ to any point in $\cball{p}{\tilde\delta}$, which concludes the proof.
\end{proof}

Next, we show an equivalent lemma for $\delta$ with a fractional part $\tilde\delta\geq\half$.

\begin{lemma}
    \label{lemma:defensiveDeltaCoverVertexAttackFractionalPartLargerHalf}
    For any graph $G$ and any $\delta$ with $\tilde\delta \geq \half$, it holds that 
    \begin{enumerate}
        \item for any point $p = p(u,v,\lambda) \in P(G)$, we have 
        $$\cball{p}{\delta} \cap V(G) \subseteq \cball{p(u,v,\half)}{\delta} \cap V(G).$$
        \item for any two points $p = p(u,v,\half), q = p(u,v,\lambda) \in P(G)$ with $1-\tilde\delta \leq \lambda \leq \tilde\delta$, we have $$\cball{p}{\delta} \cap V(G) = \cball{q}{\delta} \cap V(G).$$
    \end{enumerate}
\end{lemma}
\begin{proof}
    We first prove the first statement.
    Consider an arbitrary vertex $v^*$ that is $\delta$-covered by $p$.
    Then, $v^*$ either has distance at most $\lfloor\delta - 1 + \lambda\rfloor \leq \lfloor\delta\rfloor$ to $v$ or distance at most $\lfloor\delta - \lambda\rfloor \leq \lfloor\delta\rfloor$ to $u$.
    Since $\tilde\delta\geq\half$, it follows that $v^*$ has distance at most $\delta$ to $p(u,v,\half)$.

    For the second statement, we consider an arbitrary vertex $v^*$ that is $\delta$-covered by $q$.
    By the same argument as for the first case, $v^*$ must have distance at most $\lfloor\delta\rfloor$ to $u$ or $v$ and thus also distance at most $\delta$ to $p$.
    Now consider an arbitrary vertex $v^{**}$ that is $\delta$-covered by $p$.
    Then $v^{**}$ has distance at most $\lfloor\delta\rfloor$ to $u$ or $v$ and thus also distance at most $\delta$ to any point in $\cball{p}{\tilde\delta-\half}$, which concludes the proof.
\end{proof}

At last, we consider the special cases for $\tilde\delta = 0$ and $\tilde\delta = \half$.

\begin{lemma}
    \label{lemma:defensiveDeltaCoverVertexAttackFractionalPartZeroOrHalf}
    For any graph $G$ and any $\delta$ with $\tilde\delta = 0$ or $\tilde\delta = \half$, it holds that for any two points $p = p(u,v,\lambda_1), q = p(u,v,\lambda_2) \in P(G)$ with $0 < \lambda_1 \leq \lambda_2 < \half$ we have $\cball{p}{\delta} \cap V(G) = \cball{q}{\delta} \cap V(G)$.
\end{lemma}
\begin{proof}
    Consider an arbitrary vertex $v^*$ that is $\delta$-covered by $p$.
    Then, $v^*$ has distance at most $\lfloor \delta-\lambda_1 \rfloor$ to $u$ or at most $\lfloor \delta-1+\lambda_1 \rfloor$ to $v$.
    For the first case, let $\tilde\delta = 0$.
    Since $0 < \lambda_1 \leq \lambda_2 < \half$, we have $\lfloor \delta-\lambda_1 \rfloor = \lfloor \delta-\lambda_2 \rfloor = \delta - 1$.
    Further, we have $\lfloor \delta-1+\lambda_1 \rfloor = \lfloor \delta-1+\lambda_2 \rfloor = \delta-1$.
    For the second case, let $\tilde\delta = \frac{1}{2}$.
    Since $0 < \lambda_1 \leq \lambda_2 < \half$, we have $\lfloor \delta-\lambda_1 \rfloor = \lfloor \delta-\lambda_2 \rfloor = \lfloor\delta\rfloor$.
    Further, we have $\lfloor \delta-1+\lambda_1 \rfloor = \lfloor \delta-1+\lambda_2 \rfloor = \lfloor\delta\rfloor-1$, since $\tilde\delta+\lambda_1 \leq \tilde\delta+\lambda_2 < \half + \half = 1$.
    In conclusion, in both cases, $v^*$ is covered by both $p$ and $q$.
\end{proof}

With these three lemmas, we can show that it is always possible to find an $\ell$-defense for {\sc Defensive $\delta$-Covering With Vertex Attack} or {\sc Defensive $\delta$-Covering With Vertex Attack And Multiset Defense} that has a very regular structure.
In particular, we call an $\ell$-defense $D$ \emph{neat}, if for every $p = p(u,v,\lambda) \in D$ for some vertices $u,v \in V(G)$ and $\lambda \leq \half$ such that if $\ell = 1$, then $\lambda \in \{0, \half\}$, and if $\ell > 1$
\begin{itemize}
    \item for $\delta$ with $0 < \tilde\delta < \half$ it holds that $\lambda = \frac{i}{\ell-1}\cdot\tilde\delta$ for $i \in \{0, \dots, \ell-1\}$.
    \item for $\delta$ with $\tilde\delta > \half$ it holds that $\lambda = \half - \frac{i}{\ell-1}\cdot(\tilde\delta-\half)$ for $i \in \{0, \dots, \ell-1\}$. 
    \item for $\delta$ with  $\tilde\delta \in \{0, \frac{1}{2}\}$, it holds that $\lambda = \frac{i}{2(\ell-1)}$ for $i \in \{0, \dots, \ell-1\}$.
\end{itemize}
Further, we call an $\ell$-defense $D$ \emph{nice}, if for every $p = p(u,v,\lambda) \in D$ for some vertices $u,v \in V(G)$ and $\lambda \leq \half$ such that 
\begin{itemize}
    \item for $\delta$ with $\tilde\delta < \half$ it holds that $\lambda = 0$.
    \item for $\delta$ with $\tilde\delta\geq\half$ it holds that $\lambda = \half$.
\end{itemize}
Thus any defense that is nice is also neat.

From \Cref{lemma:defensiveDeltaCoverVertexAttackFractionalPartSmallerHalf,lemma:defensiveDeltaCoverVertexAttackFractionalPartLargerHalf,lemma:defensiveDeltaCoverVertexAttackFractionalPartZeroOrHalf} we can directly infer the following statements, which will be useful for showing containment in various complexity classes.

\begin{corollary}\label{corollary:defensiveDeltaCoverVertexAttackNeatNice}
    The following statements hold.
    \begin{enumerate}
        \item For {\sc Defensive $\delta$-Covering With Vertex Attack}, if a graph $G$ admits an $\ell$-defense that counters every $k$-attack, then $G$ also admits a neat $\ell$-defense that counters every $k$-attack.
        \item For {\sc Defensive $\delta$-Covering With Vertex Attack And Multiset Defense}, if a graph $G$ admits an $\ell$-defense that counters every $k$-attack, then $G$ also admits a nice $\ell$-defense that counters every $k$-attack.
    \end{enumerate}
\end{corollary}

\begin{corollary}
    \label{corollary:defensiveDeltaCoverVertexAttackNeatNiceMultiset}
    For any $\delta$ with $\tilde\delta \neq 0$ and $\tilde\delta \neq \half$, the minimum sizes of a nice defense that is a multiset and a neat defense that is not a multiset are equal.
\end{corollary}

\subsection*{Complexity Results}

We start by analyzing the complexity of {\sc Defensive $\delta$-Covering With Vertex Attack} when the attacker may attack each vertex only once.
There are three ranges for $\delta$ with a different resulting complexity for {\sc Defensive $\delta$-Covering With Vertex Attack}:
If $\delta < \half$, the problem is polynomial-time solvable;
if $\half \leq \delta < 1$, the problem is \NP-complete for $k \geq 2$ and polynomial time solvable for $k = 1$;
if $1 \leq \delta$, the problem is $\Sigma^P_2$-complete.

\paragraph*{The Case $\delta < \frac{1}{2}$}

\begin{lemma}\label{lemma:defensiveDeltaCoverVertexAttackSmallerHalfInP}
    {\sc Defensive $\delta$-Covering With Vertex Attack} is in $\sf P$ for any $\delta < \frac{1}{2}$ regardless of whether the defense may be a multiset or not.
\end{lemma}
\begin{proof}
    Because $\delta < \frac{1}{2}$, each token of the defender is able to defend only its nearest vertex.
    The attacker can attack each vertex at most once.
    Thus placing a defender's token on each vertex is enough.
    If one of the vertices is not defended, the attacker can attack this very vertex.
    Accordingly, the defender has to place exactly one token within distance $\delta$ of each vertex.
    Thus for any graph $G$ a smallest $\ell$-defense that counters every $k$-attack for {\sc Defensive $\delta$-Covering with Vertex Attack} and {\sc Defensive $\delta$-Covering with Vertex Attack And Multiset Defense} is to place one token on each element of $V(G)$.
\end{proof}

\paragraph*{The Case $\frac{1}{2} \leq \delta < 1$}

This case is related to the question whether, given graphs $G$ and $H$, $G$ has a $H$-factor.
A \emph{factor} of a graph $G$ is a spanning subgraph of $G$ with the same vertex set $V(G)$.
For an arbitrary graph $H$, an $H$-factor of $G$ is a factor of $G$ consisting of disjoint copies of $H$.
Given connected graphs $G$ and $H$, deciding whether there exists an $H$-factor of $G$ is {\sf NP}-complete if and only if $H$ has a component with at least three vertices \cite{DBLP:journals/siamcomp/KirkpatrickH83}.
Further, for a family of graphs $\mathcal H$, an $\mathcal H$-factor of $G$ is a factor of $G$ where each connected component of the factor is a member of $\mathcal H$.
In the following, we use the family of trees with at least $k$ edges, which we denote by $\mathcal T_k$.

\begin{theorem}\label{thm:defensiveDeltaCoverVertexAttackSmallerOneCompleteness}
    {\sc Defensive $\delta$-Covering with Vertex Attack} is $\NP$-complete for any $\frac{1}{2} \leq \delta < 1$ regardless of whether the defense may be a multiset or not.
\end{theorem}
\begin{proof}
    While \Cref{corollary:defensiveDeltaCoverVertexAttackNeatNice} allows us to show containment in $\Sigma^P_2$ easily, we need a slightly more elaborate argument to show containment in \NP.
    For that, we use the following statement, which we will also use for the correctness of the \NP-hardness proof.

    \begin{numberedclaim}
        \label{claim:graphFactorDefenseEquivalence}
        For any connected graph $G$, $\frac{1}{2} \leq \delta < 1$ and $k \leq \ell < |V(G)|$, the following two statements are equivalent:
        \begin{enumerate}
            \item $G$ admits an $\ell$-defense that counters every $k$-attack.
            \item $G$ has a $\mathcal T_k$-factor with at most $\ell$ edges.
        \end{enumerate}
    \end{numberedclaim}
    \begin{claimproof}
        We first show $1. \Rightarrow 2.$:
        Let $D$ be an $\ell$-defense for $G$ that counters every $k$-attack.
        We define the following multigraph $G'$ with $V(G') = V(G)$, initially setting $E(G') = \emptyset$.
        For each token $D$ places on an edge $e$, we add a copy of $e$ to $E(G')$.
        For each token $D$ places on a vertex $v$, we add an edge $\{v, v\}$ to $E(G')$.
        Now we consider each component $C$ of $G'$ that is not a tree.
        If $C$ contains at least $k+1$ vertices, we replace the edges in $G'[V(C)]$ by a spanning tree on $C$, reducing the total number of edges in $E(G')$.
        Otherwise, there exists an edge $e$ in $E(G)$ that connects $C$ to some other component $C'$ of $G'$. 
        We replace the edges of $G'[V(C)]$ by a spanning tree on $C$ and connect $C$ to $C'$ using $e$.
        Since the spanning tree on $C$ has exactly one fewer edge than $C$ has vertices, this in total either reduces the total number of edges in $E(G')$ or leaves it unchanged.
        Thus after these modifications, $G'$ contains no cycle and has at most $\ell$ edges in total.
        Further, there is no component $C$ of $G'$ with at most $k$ vertices: Since neither modification can create a tree with at most $k$ vertices, it must have already existed initially.
        Consider an attack exactly on the vertices of $C$.
        As $\delta < 1$, no defender token in $P(G \setminus C)$ can defend the vertices of $C$.
        By definition of $G'$, there are no defender tokens on the edges between $C$ and $V(G) \setminus V(C)$.
        Further, there are $|V(C)|-1 \leq k-1$ defender tokens in total on the edges of $C$, because $C$ is a tree, thus at least one vertex of $C$ cannot be defended.
        This is a contradiction to $D$ being an $\ell$-defense that counters every $k$-attack.
        Thus the modified graph $G'$ is a $\mathcal T_k$-factor of $G$.

        Now we show $2. \Rightarrow 1.$:
        Let $G'$ be a $\mathcal T_k$-factor of $G$ with at most $\ell$ edges.
        Define $D = \{p(u,v,\half) \mid \{u,v\} \in E(G')\}$.
        Consider any $k$-attack $A$.
        Then $A$ contains at most $k$ vertices of any component $C$ of $G'$.
        We then root $C$ at an arbitrary vertex not contained in $A$ and match the defender token on each edge to the child in the tree it is incident to.
        Thus $D$ defends $A$, and thus $D$ is an $\ell$-defense.
    \end{claimproof}

    The $\mathcal T_k$-factor problem is in \NP, since we can just use the edges of the $\mathcal T_k$-factor as the certificate.
    We verify the certificate by checking that no connected component induced by those edges contains a cycle in polynomial time and each component has at least $k+1$ vertices.
    To obtain a certificate for a {\sc Defensive $\delta$-Covering with Vertex Attack} instance, use the $\mathcal T_k$-factor certificate for each connected component of the instance graph.
    By \Cref{claim:graphFactorDefenseEquivalence}, {\sc Defensive $\delta$-Covering with Vertex Attack} is contained in \NP{} for $\frac{1}{2} \leq \delta < 1$, as cases with $\ell \geq |V(G)|$ are trivial yes-instances.

    By $P_3$ we denote the path on $3$ vertices.
    We now show that {\sc Defensive $\delta$-Covering with Vertex Attack} is \NP-hard by reduction from the $P_3$-factor problem on connected graphs.
    Given a graph $G$, if $3$ does not divide $|V(G)|$ without remainder, it is a trivial no-instance and we map to a trivial no-instance of {\sc Defensive $\delta$-Covering with Vertex Attack}.
    Otherwise, we map it to $G, k = 2, \ell = \frac{2}{3}|V(G)|$, which can be done in polynomial time.
    
    For the correctness, assume $G$ has a $P_3$-factor. 
    Since $P_3$ is a path with two edges and three vertices, a $P_3$-factor of a graph $G$ must contain exactly $\frac{2}{3}|V(G)|$ edges.
    Thus by \Cref{claim:graphFactorDefenseEquivalence}, $G$ admits a $\frac{2}{3}|V(G)|$-defense for $k = 2$, that counters every $2$-attack.

    On the other hand, if $G$ admits a $\frac{2}{3}|V(G)|$-defense that counters every $2$-attack, by \Cref{claim:graphFactorDefenseEquivalence}, $G$ has a $\mathcal T_2$-factor with at most $\frac{2}{3}|V(G)|$ edges.
    By the pigeonhole principle, this means that every tree in that $\mathcal T_2$-factor must have exactly $2$ edges, and thus that $\mathcal T_2$-factor is also a $P_3$-factor.

    Since the proof of \Cref{claim:graphFactorDefenseEquivalence} does not change if the defense is allowed to be a multiset, we also get \NP-completeness for {\sc Defensive $\delta$-Covering with Vertex Attack And Multiset Defense}.
\end{proof}

In the above proof, we showed \NP-hardness for $k = 2$, and thus \NP-hardness if $k$ is in the input, which is the standard definition.
However, we want to note that {\sc Defensive $\delta$-Covering with Vertex Attack} for $\frac{1}{2} \leq \delta < 1$ and $k = 1$ is equivalent to {\sc Edge Cover}.
Given a graph $G$, {\sc Edge Cover} asks for a minimum-cardinality set of edges that is incident to all vertices of $G$, which is solvable in polynomial time \cite{schrijver2003combinatorial}.

\paragraph*{The Case $\delta \geq 1$}

In the last case of $\delta \geq 1$, {\sc Defensive $\delta$-Covering with Vertex Attack} behaves similarly to {\sc Distance-$d$ Defensive Dominating Set}.
We establish the following theorem, which we prove in the following \Cref{lem:defensiveDeltaCoverVertexAttackMultisetContainment,lem:defensiveDeltaCoverVertexAttackContainment,lem:defensiveDeltaCoverVertexAttackHardness} and \Cref{cor:defensiveDeltaCoverVertexAttackMultisetDefenseHardness}.

\begin{theorem}\label{thm:defensiveDeltaCoverVertexAttackGreaterOneCompleteness}
\label{thm:defensiveDeltaCoverVertexAttackMultisetDefenseCompleteness}
    {\sc Defensive $\delta$-Covering with Vertex Attack} and {\sc Defensive $\delta$-Covering with Vertex Attack and Multiset Defense} are $\Sigma^P_2$-complete for any $\delta \geq 1$.
\end{theorem}

We start the proof of \Cref{thm:defensiveDeltaCoverVertexAttackGreaterOneCompleteness} by establishing the containment of {\sc Defensive $\delta$-Covering with Vertex Attack} and {\sc Defensive $\delta$-Covering with Vertex Attack and Multiset Defense} in $\Sigma^P_2$ using the structural observations that we have made earlier in \Cref{lemma:defensiveDeltaCoverVertexAttackFractionalPartSmallerHalf,lemma:defensiveDeltaCoverVertexAttackFractionalPartLargerHalf,lemma:defensiveDeltaCoverVertexAttackFractionalPartZeroOrHalf}.

\begin{lemma}\label{lem:defensiveDeltaCoverVertexAttackMultisetContainment}
    {\sc Defensive $\delta$-Covering with Vertex Attack and Multiset Defense} is in $\Sigma^P_2$ for any $\delta \geq 1$.
\end{lemma}
\begin{proof}
    We provide a polynomially-sized $\exists\forall$-witness and start with the $\exists$-quantified part, in which we encode the $\ell$-defense.
    We first argue that every defense can be polynomially encoded.
    By \Cref{corollary:defensiveDeltaCoverVertexAttackNeatNice}.2, if there is an $\ell$-defense $D$ in graph $G$, then there is also a nice $\ell$-defense $D'$.
    We encode such a nice defense $D'$ as a multiset by specifying for each possible element its multiplicity.
    Since the attacker is only able to attack each vertex once, it is enough to place a maximum of $|V(G)|$ defender tokens.
    Thus, the multiset $D'$ is always polynomially encodable.
    With the $\forall$-quantified part, we encode all possible $k$-attacks, which are subsets of vertices. 
    
    We verify the $\exists\forall$-witness by constructing a bipartite graph consisting of the defender's tokens on one side and the attacked vertices on the other.
    An edge exists if and only if a defender token is able to reach an attacked vertex.
    A matching in this graph corresponds to moving the defenders to the attacked vertices.
    If and only if the maximum matching has size at most $k$, the attack can be defended and we have a yes-instance.
\end{proof}

\begin{lemma}\label{lem:defensiveDeltaCoverVertexAttackContainment}
    {\sc Defensive $\delta$-Covering with Vertex Attack} is in $\Sigma^P_2$ for any $\delta \geq 1$.
\end{lemma}
\begin{proof}
    We first consider the two cases $0 < \tilde\delta < \frac{1}{2}$ and $\frac{1}{2} < \tilde\delta < 1$, followed by $\tilde\delta \in \{0,\frac{1}{2}\}$.
    We use the same witness and verifier as in \Cref{lem:defensiveDeltaCoverVertexAttackMultisetContainment}.
    We argue why each $\ell$-defense $D$ can be encoded in polynomial size.
    Again, since each vertex can be attacked at most once, $|V(G)|$ defenders are always enough to defend any possible attack.
    
    For $0 < \tilde\delta < \frac{1}{2}$, we encode the $\ell$-defense $D$ as a multiset of vertices.
    By \Cref{corollary:defensiveDeltaCoverVertexAttackNeatNice}.1 and \Cref{lemma:defensiveDeltaCoverVertexAttackFractionalPartSmallerHalf}, it is possible to move the defender tokens from the vertices on the incident edges to obtain a neat set $\ell$-defense such that each defender token covers the same set of vertices.
    
    For $\frac{1}{2} < \tilde\delta < 1$, we encode the $\ell$-defense $D$ as a multiset of edges.
    By \Cref{corollary:defensiveDeltaCoverVertexAttackNeatNice}.1 and \Cref{lemma:defensiveDeltaCoverVertexAttackFractionalPartLargerHalf}, it is possible to move the defender tokens from the midpoint of the edge to obtain a neat set $\ell$-defense such that each defender token covers the same set of vertices.
    
    In the cases $\tilde\delta \in \{0,\frac{1}{2}\}$, we also use \Cref{corollary:defensiveDeltaCoverVertexAttackNeatNice}.1.
    In the case $\tilde\delta=0$, by \Cref{lemma:defensiveDeltaCoverVertexAttackFractionalPartSmallerHalf}, we can group the defender tokens on an edge $\{u,v\}$ into three groups such that all defender tokens within the same group can defend the same set of vertices:
    defender token on $u$, defender token on $v$, defender tokens on $p(u,v,\lambda)$ with $\lambda \in (0,1)$.
    In the case $\tilde\delta=\frac{1}{2}$, by \Cref{lemma:defensiveDeltaCoverVertexAttackFractionalPartLargerHalf}, we can group the defender tokens on an edge $\{u,v\}$ into three groups such that all defender tokens within the same group can reach the same set of vertices:
    defender token on $p(u,v,\frac{1}{2})$, defender tokens that are on $p(u,v,\lambda)$ with $\lambda \in [0,\frac{1}{2})$, and defender tokens that are on $p(u,v,\lambda)$ with $\lambda \in (\frac{1}{2},1]$.
    In both cases, we can build an $\exists\forall$-witness by encoding the neat defense $D$ using the above groups per edge as multisets.

    In all cases, since each vertex can be attacked at most once, $|V(G)|$ defender tokens are always enough to defend any possible attack.
    Therefore, the defense $D$ can be encoded polynomially.
    We now encode all possible $k$-attacks, which are subsets of vertices, with the $\forall$-quantified part.
    The verifier checks whether there is a matching in the bipartite graph based on the defense and attack.
\end{proof}

Now, we prove the $\Sigma^P_2$-hardness of {\sc Defensive $\delta$-Covering with Vertex Attack} and {\sc Defensive $\delta$-Covering with Vertex Attack and Multiset Defense}.

\begin{lemma}\label{lem:defensiveDeltaCoverVertexAttackHardness}
    {\sc Defensive $\delta$-Covering with Vertex Attack} is $\Sigma^P_2$-hard for any $\delta \geq 1$.
\end{lemma}
\begin{proof}
    For the hardness, we distinguish between three cases:
    \begin{enumerate}
        \item $\delta \in [1, \frac{3}{2})$
        \item $\delta \in [\frac{3}{2},2)$ 
        \item $\delta \geq 2$ with subcases $\tilde\delta < \frac{1}{2}$ and $\tilde\delta \geq \frac{1}{2}$
    \end{enumerate}

    \proofsubparagraph{The Case $\delta \in [1, \frac{3}{2})$}
    We use the instance $(G,k,\ell)$ of \Cref{lem:homevertex} without modification.
    Since $\delta < \frac{3}{2}$, a token on a vertex $v$ is able to defend exactly $N[v]$.
    Therefore, the defense of~\cite{DBLP:conf/mfcs/ChaplickGK25} is a set defense that counters every $k$-attack whenever the original instance is a yes-instance.
    Conversely, consider a token at $p(u,v,\lambda)$ with $\lambda \leq \frac{1}{2}$.
    It reaches a vertex $w \not\in \{u,v\}$ only if $\lambda + dist(u,w) \leq \delta$, as $1 - \lambda + dist(v,w) \geq \frac{3}{2} > \delta$.
    Thus $w \in N[u]$, so $u$ is a home of this token, and Lemma~\ref{lem:homevertex} completes this case.

    \proofsubparagraph{The Case $\delta \in [\frac{3}{2},2)$}
    For the second case, we, unfortunately, cannot use a reduction utilizing subdivision of edges since $\delta < 2$.
    Therefore, we adapt the reduction by \cite{DBLP:conf/mfcs/ChaplickGK25} from {\sc Minimum Cardinality Clique Interdiction} to {\sc Defensive Domination} such that it reduces to {\sc Defensive $\delta$-Covering with Vertex Attack} for $\delta \in [\frac{3}{2}, 2)$.
    The overall construction and correctness proof are very similar, but differ in technicalities.

    Let $(G, s, t)$ be a {\sc Minimum Cardinality Clique Interdiction} instance, where $G$ is a graph, $s$ is the interdiction budget and $t$ is the clique size threshold.
    We assume w.l.o.g that
    \begin{itemize}
        \item $t \geq 4$, otherwise set $t=4$ while adding $s+1$ independent $4-t$-cliques and adding edges between them and all vertices in $V(G)$;
        \item $\binom{t}{2} \leq |V(G)|+s$, otherwise add isolated vertices to $G$.
    \end{itemize}
    We construct a {\sc Defensive $\delta$-Covering with Vertex Attack} instance $(G', k, \ell)$ as follows, with $V(G')$ partitioned into the following sets, see also \Cref{fig:deltaCover3over2}:
    \begin{description}
        \item[$W$] For each vertex $v \in V(G)$, we introduce four vertices $v'_L, v'_R, v''_L, v''_R$.
        That is $W = W_L \cup W_R$ with $W_L=\{v'_L,v''_L : v\in V(G)\}$ and $W_R=\{v'_R,v''_R : v\in V(G)\}$.
        Furthermore, we introduce two edges $e'_v = \{v'_L, v'_R\}$, $e''_v = \{v''_L, v''_R\}$ for every $v \in V(G)$.
        We denote this newly introduced matching by $M(W)$.
        
        \item[$F$] For each edge $e = \{u, v\} \in E(G)$, we introduce a vertex $f_e$, i.e. $F = \{f_e : e \in E(G)\}$, and four edges $\{f_e, u'_L\}, \{f_e, u''_L\}, \{f_e, v'_L\}, \{f_e, v''_L\}$.
        
        \item[$I_V$] For each vertex $v$, we introduce an independent set $I_{v}$ of size $\binom{t}{2}$.
        The union of $I_{v}$ over all $v \in V(G)$ is ${I_V}$.
        We connect all vertices in $I_v$ to $v'_R$.
        
        \item[$I'_V$] For each vertex $v$, we introduce two independent sets $I'_{v,L}$ and $I'_{v,R}$ of size $t$.
        We define $I'_v = I'_{v,L} \cup I'_{v,R}$.
        We join $I'_{v,L}$ by exactly one edge in $I'_{v,R}$, i.e. we connect each $u \in I'_{v,L}$ to exactly one $u' \in I'_{v,R}$ with edge $\{u,u'\}$.
        The union of $I'_{v,L}$ (resp. $I'_{v,R}$) over all $v \in V(G)$ is $I'_{V,L}$ (resp. $I'_{V,R}$).
        The union of $I'_{V,L}$ and $I'_{V,R}$ is $I'_V$.
        We connect the vertices of $I_v$ and $I'_{v,L}$ to a complete bipartite graph.
        
        \item[$I_1, I_4, I_5$] For $j \in \{1,4,5\}$, $I_j$ is an independent set of size $|V(G)|+s$.
        
        \item[$Q_1, Q_5$] For $j \in \{1,5\}$, $Q_j$ consists of two cliques $Q_{j,L}$ and $Q_{j,R}$ of size $|V(G)|+s$.
        Additionally, each vertex of $Q_{j,L}$ is connected to exactly one vertex in $Q_{j,R}$.
        We denote the newly introduced matching by $M(Q_j)$.
        We connect the independent set $I_j$ to all vertices $Q_{j,L}$.
        We connect $Q_{5,R}$ to all vertices in $I'_{V,L}$.

        \item[$Q_4$] It consists of two cliques $Q_{4,L}$ and $Q_{4,R}$ of size $|V(G)|+s$.
        Additionally, each vertex of $Q_{4,L}$ is connected to exactly one vertex in $Q_{4,R}$.
        We denote the newly introduced matching by $M(Q_4)$.
        We connect the independent set $I_4$ to all vertices $Q_{4,R}$.
        We connect $Q_{4,L}$ to all vertices in $W$.
        
        \item[$I_2$] This set is an independent set of size $|V(G)|+s-{\binom{t}{2}}$.
        
        \item[$Q_2$] This set consists of two cliques $Q_{2,L}$ and $Q_{2,R}$ of size $|V(G)|+s-(t+1)$.
        Additionally, each vertex of $Q_{2,L}$ is connected to exactly one vertex in $Q_{2,R}$.
        We denote the newly introduced matching by $M(Q_2)$.
        We connect $Q_{1,R}$ to all vertices in $Q_2$.
        We connect the independent set $I_2$ to all vertices $Q_{2,L}$.
        We connect $Q_{2,R}$ to all vertices in $F$.
        We connect $Q_{2,R}$ to all vertices in $I_V$.
        
        \item[$I_3$] This set is an independent set of size $2(|V(G)|+s)$.
        We connect $W_R$ to all vertices in $I_3$.
    \end{description}
    At last, we set $k = |V(G)| + s$ and $\ell = 5(|V(G)|+s) + t \, |V(G)| - (t+1)$.

    Given that $\delta \in [\frac{3}{2},2)$, a defender token on an edge $\{u,v\}$ is able to defend $N[u] \cup N[v]$ since any other vertex has distance $\geq 2$ from both $u$ and $v$.
    A defender token on a vertex $v$ is able to defend $N[v]$ for the same reason.
    We therefore, define a defense $D$ as a set of vertices and edges, on which we place the token at the midpoint.

    In the following we prove that $s$ vertices can be deleted from $G$ such that no clique of size $t$ remains if and only if there is an $\ell$-defense that counters every $k$-attack in $G'$ (\Cref{fig:deltaCover3over2}).
    We let $E(X,Y):=\{\{x,y\}\in E(G'):x\in X,\ y\in Y\}$.

    We start with the forward direction.
    For this, let $(G,s,t)$ be a yes-instance.
    That is there is a blocker $B$ with $|B| = s$ such that $G \setminus B$ does not contain a clique of size at least $t$.
    Based on this, we construct a defense $D$ as follows.
    $$
        D = M(Q_1) \cup M(Q_2) \cup M(Q_4) \cup M(Q_5) \cup \bigcup_{v \in V(G)} T_{I_v, I'_v} \cup \{e'_v : v \in V(G)\} \cup \{e''_v : v \in B\},
    $$
    where $T_{I_v, I'_v} \subseteq \{\{u,u'\} : u \in I_v, u'\in I'_{v,L}\}$ is a matching of size $t$.

    We show that $D$ counters every $k$-attack.
    For contradiction, we assume that there is an inclusion-minimal $k$-attack $|A|$ with $count_D(A) < |A|$.
    All vertices $v$ in the independent sets $I_1, I_3, I_4, I_5$, the sets $Q_1, Q_2, Q_4, Q_5$, and set $W$ are defended, since all of them have $|V(G)|+s$ defenders in $B^{\leq}(v, \delta)$.
    Therefore, $count_D(A) \geq |V(G)| + s \geq |A|$.
    It follows that $A \subseteq I_2 \cup F \cup I_V \cup I'_V$.
    
    Moreover, the vertices in the left part $I'_{V,L}$ of $I'_V$ are defended by $M(Q_5)$.
    If any of these vertices is part of $A$, we have $count_D(A) \geq |V(G)| + s \geq |A|$.
    Further, if a vertex of $I_V$ is attacked, we have $count_D(A) \geq |M(Q_2)| + |T_{I_v, I'_v}| + |\{e_v'\}| = (|V(G)|+s-(t+1)) + t + 1 = |V(G)| + s$ for some $v \in V(G)$.
    The right side $I'_{V,R}$ of the vertices in $I'_V$ can be directly defended by $e \in T_{I_v, I'_v}$ because $B^\leq(e,\delta) \setminus I'_{V,R} \subseteq I'_{V,L} \cup I_V \cup Q_2 \cup Q_5 \cup W$, which are not attacked as described above.
    By minimality of $A$, we have $A \cap I'_{V,R} = \emptyset$.
    Thus, $A \subseteq I_2 \cup F$.

    Because, we assume $t \geq 4$, we have $|I_2| \leq |V(G)|+s-(t+1)$.
    Since there are $|V(G)|+s-(t+1)$ in $M(Q_2)$, we would have $count_D(A) \geq |A|$, hence $A \not\subseteq I_2$.
    Therefore, $A \cap F \neq \emptyset$ and since each vertex in $F$ is adjacent to at least two defenders in $D$
    $$
        |A \cap F| \geq |A| - |I_2| > |V(G)| + s - (t+1) + 2 - (|V(G)|+s-{\binom{t}{2}}) = {\binom{t-1}{2}}
    $$
    It follows that $|B^\leq(A \cap F, \delta) \cap D| \geq t$, hence every minimal attack has size $|A| = |V(G)| + s$ since otherwise $count_D(A) \geq |V(G)| + s - 1 \geq |A|$.
    
    The result is an attack $A$ with $|A \cap F| \geq {\binom{t}{2}}$.
    A subset of ${\binom{t}{2}}$ edges from $E(G)$ in $G$ is either incident to at least $t+1$ vertices from $V(G)$ or induces a clique of size $t$ in $G$.
    We further have $count_{D \cap M(Q_2)}(A) = |V(G)| + s - (t+1)$ and $M(Q_2) \cap M(W) = \emptyset$.
    In the first case, where $A \cap F$ is incident to at least $t+1$ vertices in $G$, $M(W)$ contains at least $t+1$ defenders and thus $count_D(A) \geq |V(G)| + s$.
    In the second case, $M(W)$ contains also at least $t+1$ defenders since $e''_v \in D$ for some $v$ in the clique and thus $count_D(A) \geq |V(G)| + s$.

    For the backward direction, we assume that $(G', k, \ell)$ is a yes-instance, i.e. there is a defense $D$ of size at most $\ell$ that counters every attack of size at most $k$.
    $D$ needs to counter every attack, particularly
    \begin{enumerate}
        \item For $j \in \{1,4,5\}$, $|V(G)|+s$ defenders are placed in $B^\leq(I_j, \delta)$, since every attack on $I_j$ needs to be defended.
        \item For each $v \in V(G)$, $t$ defenders need to be placed in $B^\leq(I'_{v,R}, \delta)$.
        \item In $B^\leq(I_2, \delta)$, there are at least $|V(G)| + s - \binom{t}{2}$ defenders.
        \item In $B^\leq(I_3, \delta)$ at least $|V(G)|+s$ defenders need to be placed.
    \end{enumerate}
    The sets in items 1--4 are pairwise disjoint, so they account for $5(|V(G)|+s) + t\,|V(G)| - \binom{t}{2}$ of the $\ell$ tokens.
    We construct a nice $\ell$-defense $D'$ from $D$ as follows in accordance with \Cref{lemma:defensiveDeltaCoverVertexAttackFractionalPartLargerHalf} using the following exchange arguments, which we apply in the given order.
    For $j \in \{1,4,5\}$, we place $|V(G)| + s$ defenders on the matching edges $M(Q_j)$.
    Thus, all vertices of $I_j$ and $Q_j$ can be defended.
    Because $\delta < 2$, every other placement in $B^\leq(I_j,\delta)$ reaches only the vertices reached by the tokens on $M(Q_j)$.
    
    For each $v \in V(G)$, we place $t$ defenders on edges $M_v \subseteq E_{I_v, I'_v}$ such that $M_v$ is a matching between $I_v$ and $I'_{v,L}$.
    Since $I'_{v,L}$ and $I_v$ build up a complete bipartite graph, we can cover all vertices in $I'_v \cup I_v$.
    We further are able to additionally defend also vertices in $W_R$ and $Q_2$; placing the defenders within $I'_V$ would not give us this opportunity.
    
    Since $Q_1$ and $Q_2$ are defended, the $|V(G)|+s-\binom{t}{2}$ defender tokens of $B^\leq(I_2, \delta)$ can be placed on $M(Q_2)$ so they additionally reach $I_2,Q_2,I_V,$ and $F$.

    There are six options for placing $|V(G)| + s$ defender tokens in $B^\leq(I_3, \delta)$: $I_3$, $W_R$, $M(W)$, $E(W,I_3)$, $E(W, I_V)$, $E(W,Q_4)$.
    Placing the defenders on $I_3$ is not defending all vertices in $I_3$.
    Placing the defender tokens on $W_R$ is also not optimal, since $\delta \geq \frac{3}{2}$ and placing them on the incident edges defends at least $B^\leq(W_R, \delta)$.
    The vertices in $Q_4, I'_{V,L}, Q_{2,R}$ are already fully defended since $count_{D'}(A) \geq |V(G)|+s$.
    Therefore, the dominant options $M(W)$ and $E(W,I_V)$ remain.

    Every vertex in $Q_1 \cup Q_2 \cup Q_4 \cup Q_5 \cup I_1 \cup I_3 \cup I_4 \cup I_5 \cup W \cup I'_{V,L}$ is within reach of at least $|V(G)|+s$ tokens of $D'$, so $count_{D'}(A) \ge |A|$ for every $k$-attack $A$ containing such a vertex.
    Each vertex of $I'_{V,R}$ has a correspondent token in $M_v$, so an inclusion-minimal successful attack $A$ satisfies $A \subseteq I_2 \cup F \cup I_V$.
    Furthermore, for every $v \in V(G)$, there are at least $|V(G)|+s$ defenders in $B^\leq(I_2 \cup I_v, \delta)$, since $|I_2 \cup I_v| = |V(G)|+s$.
    Since $|V(G)|+s-\binom{t}{2}$ defender tokens are already placed in $M(Q_2) \subseteq B^\leq(I_2, \delta)$ and $t$ defender tokens in $M_v \subseteq B^\leq(I_v, \delta)$, this leaves $\binom{t}{2} - t$ tokens that additionally need to be in $B^\leq(I_2 \cup I_v, \delta)$.
    There are currently $\binom{t}{2} - (t+1)$ unassigned tokens.
    Placing them in $M(Q_2)$ is optimal since they are in $B^\leq(I_2, \delta) \cap B^\leq(I_v, \delta) \cap B^\leq(F, \delta)$.
    Thus, the whole set $M(Q_2)$ of $|V(G)|+s-(t+1)$ tokens is in $D'$.
    We are now able to decide where to put the $|V(G)|+s$ defenders from $B^\leq(I_3, \delta)$.
    Since $v'_R \in W_R$ is adjacent to $I_v$, placing the last required defender on $e'_v \in M(W)$ fully defends $I_v$, while also defending vertices in $F$.
    
    We have a remainder of $s$ defender tokens.
    These could be placed on $M(W)$, $E(W,I_3)$, $E(W,Q_4)$, or $E(W,I_V)$.
    Again, since $Q_4,I_3,I_V,I'_{V,L}$ are defended by at least $|V(G)|+s$ defenders, the dominant choice is to place the defenders on $M(W)$.
    It remains to defend all attacks on $I_2$ and $F$ since all vertices in the other sets have at least $|V(G)|+s$ defender tokens within distance $\delta$.
    The defenders on $M(Q_2)$ reach both the sets $I_2$ and $F$.
    Accordingly, an attack includes at least ${\binom{t}{2}}$ vertices corresponding to an edge set $E' \subseteq E(G)$ such that no more than $t$ vertices in $G$ are incident to $E'$.
    If more than ${\binom{t}{2}}$ vertices in $F$ are attacked, $E'$ is incident to at least $t+1$ vertices, which is already defended by the defenders in $e'_v$ for all $v \in V(G)$.
    If exactly ${\binom{t}{2}}$ vertices in $F$ are attacked, $E'$ is either incident to at least $t+1$ vertices or $E'$ induces a clique in $G$.
    The first case can be handled as above.
    In the second case, because $D'$ is an $\ell$-defense that counters every $k$-attack, there needs to be at least two defender tokens on $e'_v$ or $e''_v$ for some $v \in V(G)$ in the clique in $G$.
    W.l.o.g. we put the second token on $e''_v$.
    Then, the blocker correspond to the tokens on $e''_v$, i.e. $B = \{v : e''_v \in D'\}$, and we have $|B| \leq s$.
    Since $s$ is the budget of the interdictor, all cliques are interdicted.
    
    In conclusion, an $\ell$-defense $D$ exists in $G'$ that counters every $k$-attack if and only if all cliques of size $t$ can be interdicted with an interdiction budget of at most $s$.
    
    \begin{figure}[!ht]
    \centering
    \scalebox{0.38}{
        \begin{tikzpicture}[
            cross/.style={%
                    path picture={%
                        \draw[thick] (path picture bounding box.south west) -- (path picture bounding box.north east);
                        \draw[thick] (path picture bounding box.north west) -- (path picture bounding box.south east);
                    }
                },
        ]
        
            \node[] (I1) at (-1.25,0) {};
            \node[] (I1label) at ($(I1)+(0,0)$) {\Huge $I_1$};
            \node[] (I1label) at ($(I1)+(-0.5,-0.75)$) {\Large $|V(G)|+s$};
            \draw[draw=black] ($(I1)+(0.75,-0.5)$) rectangle ++(5,1);
            \foreach \x in {0,1,3,4} {
                \node[circle,draw,thick,minimum size=0.75cm] (I1\x) at ($(I1)+(1.25,0)+(\x,0)$) {};
            }
            \node[] () at ($(I1)+(1.25,0)+(2,0)$) {$\cdots$};

            \draw ($(I1)+(3.25,-3.5)$) ellipse (3.25cm and 2cm);
            \node[] at ($(I1)+(-0.75,-3.5)$) {\Huge $Q_1$};
            \node[] at ($(I1)+(-1.25,-4.25)$) {\Large $|V(G)|+s$};
            \foreach \x in {0,1,3,4} {
                \node[circle,draw,thick,minimum size=0.75cm] (Q1\x1) at ($(I1)+(1.25,-2.5)+(\x,0)$) {};
                \node[circle,draw,thick,minimum size=0.75cm] (Q1\x2) at ($(I1)+(1.25,-4.5)+(\x,0)$) {};
            }
            \node[] () at ($(I1)+(1.25,-2.5)+(2,0)$) {$\cdots$};
            \node[] () at ($(I1)+(1.25,-4.5)+(2,0)$) {$\cdots$};
            \foreach \x in {0,1,3,4} {
                \draw[thick] (Q1\x1) -- (Q1\x2);
            }
            \foreach \x in {0,1,3,4} {
                \foreach \y in {0,1,3,4} {
                    \ifnum \x < \y
                        \draw (Q1\x1) edge[bend left] (Q1\y1);
                        \draw (Q1\x2) edge[bend right] (Q1\y2);
                    \fi
                }
            }

            \foreach \x in {0,1,3,4} {
                \foreach \y in {0,1,3,4} {
                    \draw[lightgray,thick] (I1\x) -- (Q1\y1);
                }
            }

            \node[] (I2) at (-2.25,-7.25) {\Huge $I_2$};
            \node[] at ($(I2)+(-2.1,-0.75)$) {\Large $|V(G)|+s-{t \choose 2}$};
            \draw[draw=black] ($(I2)+(-0.5,-0.75)$) rectangle ++(1,-5);
            \foreach \y in {0,1,3,4} {
                \node[circle,draw,thick,minimum size=0.75cm] (I2\y) at ($(I2.center)+(0,-1.25)+(0,-\y)$) {};
            }
            \node[] () at ($(I2.center)+(0,-1.2)+(0,-2)$) {$\vdots$};
            
            \draw ($(I2)+(4.25,-3.25)$) ellipse (2cm and 3.25cm);
            \node[] at ($(I2)+(4,-7.25)$) {\Huge $Q_2$};
            \node[] at ($(I2)+(4,-8)$) {\Large $|V(G)|+s-(t+1)$};
            \foreach \y in {0,1,3,4} {
                \node[circle,draw,thick,minimum size=0.75cm] (Q2\y1) at ($(I2)+(3.25,-1.25)+(0,-\y)$) {};
                \node[circle,draw,thick,minimum size=0.75cm] (Q2\y2) at ($(I2)+(5.25,-1.25)+(0,-\y)$) {};
            }
            \node[] () at ($(I2)+(1.25,-3.2)+(2,0)$) {$\vdots$};
            \node[] () at ($(I2)+(1.25,-3.2)+(4,0)$) {$\vdots$};
            \foreach \y in {0,1,3,4} {
                \draw[thick] (Q2\y1) -- (Q2\y2);
            }
            \foreach \x in {0,1,3,4} {
                \foreach \y in {0,1,3,4} {
                    \ifnum \x < \y
                        \draw (Q2\x1) edge[bend right] (Q2\y1);
                        \draw (Q2\x2) edge[bend left] (Q2\y2);
                    \fi
                }
            }
            
            \foreach \x in {0,1,3,4} {
                \foreach \y in {0,1,3,4} {
                    \draw[lightgray,thick] (I2\x) -- (Q2\y1);
                }
            }

            \foreach \x in {0,1,3,4} {
                \foreach \y in {0,1,3,4} {
                    \draw[lightgray] (Q1\x2) -- (Q2\y1);
                    \draw[lightgray] (Q1\x2) -- (Q2\y2);
                }
            }

            \node[] (I3) at (20.25,0) {\Huge $I_3$};
            \node[] at ($(I3)+(2,-0.75)$) {\Large $2(|V(G)|+s)$};
            \draw[draw=black] ($(I3)+(-0.5,-0.75)$) rectangle ++(1,-5);
            \foreach \y in {0,1,3,4} {
                \node[circle,draw,thick,minimum size=0.75cm] (I3\y) at ($(I3.center)+(0,-1.25)+(0,-\y)$) {};
            }
            \node[] () at ($(I3.center)+(0,-1.2)+(0,-2)$) {$\vdots$};

            \node[] (I4) at (25,-8.5) {};
            \node[] at ($(I4)+(0,-6.5)$) {\Huge $I_4$};
            \node[] at ($(I4)+(0,-7.25)$) {\Large $|V(G)|+s$};
            \draw[draw=black] ($(I4)+(-0.5,-0.75)$) rectangle ++(1,-5);
            \foreach \y in {0,1,3,4} {
                \node[circle,draw,thick,minimum size=0.75cm] (I4\y) at ($(I4.center)+(0,-1.25)+(0,-\y)$) {};
            }
            \node[] () at ($(I4.center)+(0,-1.2)+(0,-2)$) {$\vdots$};
            
            \draw ($(I4)+(-4.25,-3.25)$) ellipse (2cm and 3.25cm);
            \node[] at ($(I4)+(-4,-7.25)$) {\Huge $Q_4$};
            \node[] at ($(I4)+(-4,-8)$) {\Large $|V(G)|+s$};
            \foreach \y in {0,1,3,4} {
                \node[circle,draw,thick,minimum size=0.75cm] (Q4\y1) at ($(I4)-(3.25,1.25)+(0,-\y)$) {};
                \node[circle,draw,thick,minimum size=0.75cm] (Q4\y2) at ($(I4)-(5.25,1.25)+(0,-\y)$) {};
            }
            \node[] () at ($(I4)+(-1.25,-3.2)-(2,0)$) {$\vdots$};
            \node[] () at ($(I4)+(-1.25,-3.2)-(4,0)$) {$\vdots$};
            \foreach \y in {0,1,3,4} {
                \draw[thick] (Q4\y1) -- (Q4\y2);
            }
            \foreach \x in {0,1,3,4} {
                \foreach \y in {0,1,3,4} {
                    \ifnum \x < \y
                        \draw (Q4\x1) edge[bend left] (Q4\y1);
                        \draw (Q4\x2) edge[bend right] (Q4\y2);
                    \fi
                }
            }
            
            \foreach \x in {0,1,3,4} {
                \foreach \y in {0,1,3,4} {
                    \draw[lightgray,thick] (I4\x) -- (Q4\y1);
                }
            }
            
            \node[] (IV) at (9,0) {};
            \node[] () at ($(IV)+(0,0.1)$) {\Huge $I_V$};
            \draw[draw=black] ($(IV)+(-0.75,-0.5)$) rectangle ++(1.5,-5.5);
            \draw[draw=black] ($(IV)+(-0.5,-1.75)$) rectangle ++(1,-3);
            \foreach \y in {0,2} {
                \node[circle,draw,thick,minimum size=0.75cm] (IV\y) at ($(IV.center)+(0,-2.25)+(0,-\y)$) {};
            }
            \node[] () at ($(IV.center)+(0,-1.2)+(-0,-2)$) {$\vdots$};
            \node[] () at ($(IV.center)+(0,-1.2)+(-0.2,-0)$) {$\vdots$};
            \node[] () at ($(IV.center)+(0,-1.2)+(0,-4)$) {$\vdots$};
            \node[] at ($(IV.center)+(0,-1.2)+(0.3,-0)$) {\huge $I_v$};
            \draw[decoration={brace,raise=28pt, amplitude=10pt},decorate] ([xshift=-2pt]IV2.south) -- node[left=40pt] {\huge $t \choose 2$} ([xshift=-2pt]IV0.north);

            \node[] (IprimeV) at (14,0) {};
            \node[] (IprimeVlabel) at ($(IprimeV)+(0,1)$) {\Huge $I'_V$};
            \node[] (Vl) at ($(IprimeVlabel)+(-1,-0.95)$) {\Huge $I'_{V,L}$};
            \node[] (Vr) at ($(IprimeVlabel)+(1,-0.95)$) {\Huge $I'_{V,R}$};
            \draw[draw=black] ($(IprimeV)+(-1.75,-0.5)$) rectangle ++(3.5,-5.5);
            \draw[draw=black] ($(IprimeV)+(-1.5,-1.75)$) rectangle ++(3,-3);
            \foreach \y in {-1,1} {
                \node[circle,draw,thick,minimum size=0.75cm] (IprimeV\y1) at ($(IprimeV)+(-1,-3.25)+(0,-\y)$) {};
                \node[circle,draw,thick,minimum size=0.75cm] (IprimeV\y2) at ($(IprimeV)+(1,-3.25)+(0,-\y)$) {};
            }
            \node[] () at ($(IprimeV.center)+(-1,-1.3)+(0,-1.8)$) {$\vdots$};
            \node[] () at ($(IprimeV.center)+(1,-1.3)+(0,-1.8)$) {$\vdots$};
            \node[] () at ($(IprimeV.center)+(0,-1.2)+(0,-0)$) {$\vdots$};
            \node[] () at ($(IprimeV.center)+(0,-1.2)+(0,-4)$) {$\vdots$};
            \foreach \y in {-1,1} {
                \draw[thick] (IprimeV\y1) -- (IprimeV\y2);
            }
            \node[] at ($(IprimeV.center)+(0,-1.2)+(0.675,-0)$) {\huge $I'_v$};
            \draw[decoration={brace,raise=28pt, amplitude=10pt},decorate] ([xshift=-2pt]IprimeV-12.north) -- node[right=40pt] {\huge $t$} ([xshift=-2pt]IprimeV12.south);

            \foreach \x in {0,2} {
                \foreach \y in {-1,1} {
                    \draw[lightgray,thick] (IV\x) -- (IprimeV\y1);
                }
            }
            
            \node[] (V) at (14,-16) {};
            \draw[draw=black] ($(V)+(-1.75,0.75)$) rectangle ++(3.5,7);
            \foreach \y in {0,1,3,4} {
                \node[circle,draw,thick,minimum size=0.75cm] (V\y1) at ($(V)+(-1,2.25)+(0,\y)$) {};
                \node[circle,draw,thick,minimum size=0.75cm] (V\y2) at ($(V)+(1,2.25)+(0,\y)$) {};
                \draw[draw=black] ($(V\y1)+(-0.5,-0.45)$) rectangle ++(3,0.9);
            }
            \node[] () at ($(V.center)+(-1,1.3)+(0,3)$) {$\vdots$};
            \node[] () at ($(V.center)+(1,1.3)+(0,3)$) {$\vdots$};
            \foreach \y in {0,1,3,4} {
                \draw[thick] (V\y1) -- (V\y2);
            }

            \node[] (E) at (9,-16) {\Huge $F$};
            \draw[draw=black] ($(E)+(-0.5,0.75)$) rectangle ++(1,7);
            \foreach \y in {0,1,3,5,6} {
                \node[circle,draw,thick,minimum size=0.75cm] (E\y) at ($(E.center)+(0,1.25)+(0,\y)$) {};
            }
            \node[] () at ($(E.center)+(0,1.3)+(0,2)$) {$\vdots$};
            \node[] () at ($(E.center)+(0,1.3)+(0,4)$) {$\vdots$};

            \foreach \y in {0,1,3,4} {
                \draw[ultra thick] (E3) -- (V\y1);
            }

            \foreach \x in {0,1,3,4} {
                \foreach \y in {0,1,3,4} {
                    \draw[lightgray,thick] (V\x2) -- (I3\y);
                }
            }

            \foreach \x in {0,1,3,4} {
                \foreach \y in {0,1} {
                    \draw[lightgray,thick] (V\x1) .. controls (16.5,{-18-1*\y}) .. (Q4\y2);
                    \draw[lightgray,thick] (V\x2) -- (Q4\y2);
                }
                \foreach \y in {3,4} {
                    \draw[lightgray,thick] (V\x1) .. controls (16.5,{-17-1*\y}) .. (Q4\y2);
                    \draw[lightgray,thick] (V\x2) -- (Q4\y2);
                }
            }

            \foreach \x in {4} {
                \foreach \y in {0,2} {
                    \draw[lightgray,thick] (V\x2) -- (IV\y);
                }
            }

            \foreach \x in {0,1,3,5,6} {
                \foreach \y in {0,1,3,4} {
                    \draw[lightgray,thick] (E\x) -- (Q2\y2);
                }
            }

            \foreach \x in {0,1,3,4} {
                \foreach \y in {0,2} {
                    \draw[lightgray,thick] (Q2\x2) -- (IV\y);
                }
            }

            \node[] (I5) at (8,7) {\Huge $I_5$};
            \node[] at ($(I5)+(-0.5,-0.75)$) {\Large $|V(G)|+s$};
            \draw[draw=black] ($(I5)+(0.75,-0.5)$) rectangle ++(5,1);
            \foreach \x in {0,1,3,4} {
                \node[circle,draw,thick,minimum size=0.75cm] (I5\x) at ($(I5)+(1.25,0)+(\x,0)$) {};
            }
            \node[] () at ($(I5)+(1.25,0)+(2,0)$) {$\cdots$};

            \draw ($(I5)+(3.25,-3.5)$) ellipse (3.25cm and 2cm);
            \node[] at ($(I5)+(-0.75,-3.5)$) {\Huge $Q_5$};
            \node[] at ($(I5)+(-1.25,-4.25)$) {\Large $|V(G)|+s$};
            \foreach \x in {0,1,3,4} {
                \node[circle,draw,thick,minimum size=0.75cm] (Q5\x1) at ($(I5)+(1.25,-2.5)+(\x,0)$) {};
                \node[circle,draw,thick,minimum size=0.75cm] (Q5\x2) at ($(I5)+(1.25,-4.5)+(\x,0)$) {};
            }
            \node[] () at ($(I5)+(1.25,-2.5)+(2,0)$) {$\cdots$};
            \node[] () at ($(I5)+(1.25,-4.5)+(2,0)$) {$\cdots$};
            \foreach \x in {0,1,3,4} {
                \draw[thick] (Q5\x1) -- (Q5\x2);
            }
            \foreach \x in {0,1,3,4} {
                \foreach \y in {0,1,3,4} {
                    \ifnum \x < \y
                        \draw (Q5\x1) edge[bend left] (Q5\y1);
                        \draw (Q5\x2) edge[bend right] (Q5\y2);
                    \fi
                }
            }

            \foreach \x in {0,1,3,4} {
                \foreach \y in {0,1,3,4} {
                    \draw[lightgray,thick] (I5\x) -- (Q5\y1);
                }
            }

            \foreach \x in {-1,1} {
                \foreach \y in {0,1,3,4} {
                    \draw[lightgray,thick] (IprimeV\x1) -- (Q5\y2);
                }
            }
            
            \foreach \x in {0,1,3,4} {
                \node[cross,minimum size=0.5cm] at ($(Q1\x1)+(0,-1)$) {};
                \node[cross,minimum size=0.5cm] at ($(Q2\x1)+(1,0)$) {};
                \node[cross,minimum size=0.5cm] at ($(Q4\x1)+(-1,0)$) {};
                \node[cross,minimum size=0.5cm] at ($(Q5\x1)+(0,-1)$) {};
            }
            \foreach \x in {1,3,4} {
                \node[cross,minimum size=0.5cm] at ($(V\x1)+(1,0)$) {};
            }
            \foreach \x in {-1,1} {
                \node[cross,minimum size=0.5cm] at ($(IprimeV\x1)-(2,0)$) {};
            }

            \node[fill=white] (Vlabel) at ($(V)+(0,-0.75)$) {\Huge $W$};
            \node[] (Vl) at ($(Vlabel)+(-1,0.75)$) {\Huge $W_L$};
            \node[fill=white] (Vr) at ($(Vlabel)+(1,0.75)$) {\Huge $W_R$};
            
            \node[fill=white,right=of V02.center] {\Huge $e''_w$};
            \node[fill=white,right=of V12.center] {\Huge $e'_w$};
            \node[fill=white,right=of V32.center] {\Huge $e''_v$};
            \node[fill=white,right=of V42.center] {\Huge $e'_v$};

            \pgfresetboundingbox
            \path[use as bounding box] (-5.95,-19.05) rectangle (26.07,7.52);
        \end{tikzpicture}
    }
    \caption{The reduction graph $G'$ for {\sc Defensive $\delta$-Cover with Vertex Attack} for $\delta \in [\frac{3}{2},2)$. The crosses represent defender tokens that need to be placed by the arguments of \Cref{lem:defensiveDeltaCoverVertexAttackHardness}.}
    \label{fig:deltaCover3over2}
\end{figure}

    \proofsubparagraph{The Case $\delta \geq 2$}
    Let $(G, k, \ell)$ be the instance of \Cref{lem:homevertex}.
    We construct $(G', k', \ell')$ to be the instance of {\sc Defensive $\delta$-Covering With Vertex Attack} by replacing every edge $\{u,v\}$ of $G$ by a gadget, and set $k'=k$ and $\ell'=\ell+k|E(G)|$.
    Let $d = \lfloor \delta \rfloor \geq 2$.
    The gadget depends on $\tilde\delta$.

    \begin{description}
        \item[The subcase $\tilde\delta < \frac{1}{2}$.]
        We reuse the gadgets from the reduction to {\sc Distance-$d$ Defensive Domination} (\Cref{lemma:distanceDefensiveDominationCompleteness}) in \Cref{fig:distanceDefensiveDomination:odd} and \Cref{fig:distanceDefensiveDomination:even}.
        Since all distances between vertices are integral, a token on a vertex reaches exactly the vertices within distance $d$, so $k$ tokens on $C_k$ defend all gadget vertices, while $u$ and $v$ have distance $d+1 > \delta$ from $C_k$.
        Conversely, $k$ tokens within distance $\delta$ of $I_k$ are necessary, and every point within distance $\delta$ of $I_k$ has distance at least $d+1-\tilde\delta>\delta$ from $u$ and $v$.
        
        \item[The subcase $\tilde\delta \geq \frac{1}{2}$.]
        We substitute each edge $\{u,v\}$ by the gadget presented in \Cref{fig:distanceDefensiveDomination:unattackableVertex:odd} for $\lfloor\delta\rfloor$ odd and in \Cref{fig:distanceDefensiveDomination:unattackableVertex:even} for $\lfloor\delta\rfloor$ even.
        The gadget subdivides the edge $\{u,v\}$ $\lfloor \delta \rfloor - 1$ times.
        We append a path of $\lfloor \frac{\lfloor \delta \rfloor+1}{2} \rfloor$ vertices to the middle vertex, if $\lfloor \delta \rfloor$ is even, or to the two middle vertices, if $\lfloor \delta \rfloor$ is odd.
        We connect the end vertex of the path to all vertices of a clique $C^\ell_k$ of size $k$.
        We further introduce another clique $C^r_k$ of size $k$ and connect each vertex in $C^r_k$ to exactly one vertex in $C^\ell_k$.
        We then append a path of $\lfloor \delta \rfloor - 1$ vertices by connecting the first vertex to all vertices in $C^r_k$.
        At last, we connect the end vertex of the path to an independent set of size $k$.
    
        Note that all of the vertices in $I_k$ have to be defended and thus $k$ tokens are necessary but also enough.
        The farthest position to place a defender token is on the edges between the two cliques $C^\ell_k$ and $C^r_k$, since $\lfloor\delta\rfloor$ vertices separate the defender token and the vertices from $I_k$.
        Furthermore, these tokens can defend exactly all gadget vertices except $u$ and $v$.
    
        \begin{figure}
    \begin{subfigure}{0.48\textwidth}
        \centering
        \scalebox{0.32}{
            \begin{tikzpicture}[
                cross/.style={%
                    path picture={%
                        \draw[thick] (path picture bounding box.south west) -- (path picture bounding box.north east);
                        \draw[thick] (path picture bounding box.north west) -- (path picture bounding box.south east);
                    }
                }
            ]            
                \node[circle,draw,thick,minimum size=0.75cm] (u) at (-1,8) {}; \node[] (ul) at (-1.75,8) {\Huge $u$};
                \node[circle,draw,thick,minimum size=0.75cm] (v) at (-1,4) {}; \node[] (vl) at (-1.75,4) {\Huge $v$};
                \node[circle,draw] (1) at (-1,6.75) {};
                \node[circle,draw] (2) at (-1,5.25) {};
        
                \draw[black] (u) -- (1);
                \draw[black] (v) -- (2);
                \draw[black] (1) -- (2);
        
                \node[circle,draw] (21) at (0,6) {};
                \node[circle,draw] (22) at (1,6) {};
        
                \draw[black] (1) -- (21);
                \draw[black] (2) -- (21);
                \draw[black] (21) -- (22);
    
                \node[circle,draw,minimum size=0.75cm] (Cl1) at (2.5,4.5) {};
                \node[circle,draw,minimum size=0.75cm] (Cl2) at (2.5,6) {};
                \node[circle,draw,minimum size=0.75cm] (Cl3) at (2.5,7.5) {}; \node[] (Cl) at (2.5,8.5) {\Huge $C^\ell_k$};
                \node[circle,draw,minimum size=0.75cm] (Cr1) at (5.5,4.5) {};
                \node[circle,draw,minimum size=0.75cm] (Cr2) at (5.5,6) {};
                \node[circle,draw,minimum size=0.75cm] (Cr3) at (5.5,7.5) {}; \node[] (Cr) at (5.5,8.5) {\Huge $C^r_k$};
    
                \draw[thick] (Cl1) -- (Cl2);
                \draw[thick,bend right] (Cl1) edge (Cl3);
                \draw[thick] (Cl2) -- (Cl3);
                \draw[thick] (Cr1) -- (Cr2);
                \draw[thick,bend left] (Cr1) edge (Cr3);
                \draw[thick] (Cr2) -- (Cr3);
                \draw[thick] (Cl1) -- (Cr1);
                \draw[thick] (Cl2) -- (Cr2);
                \draw[thick] (Cl3) -- (Cr3);
                \node[cross,minimum size=0.5cm] (cross1) at (4,4.5) {};
                \node[cross,minimum size=0.5cm] (cross2) at (4,6) {};
                \node[cross,minimum size=0.5cm] (cross3) at (4,7.5) {};
                
                \draw[draw=black] (10,3.5) rectangle ++(1,5); \node[] (Il) at (11.5,8) {\Huge $I_k$};
                \node[circle,draw,thick,minimum size=0.75cm] (v) at (10.5,4) {};
                \node[circle,draw,thick,minimum size=0.75cm] (v) at (10.5,5) {};
                \node[] (v) at (10.5,6) {\vdots};
                \node[circle,draw,thick,minimum size=0.75cm] (v) at (10.5,7) {};
                \node[circle,draw,thick,minimum size=0.75cm] (v) at (10.5,8) {};
                \node[] (Il1) at (10,4.75) {};
                \node[] (Il2) at (10,6) {};
                \node[] (Il3) at (10,7.25) {};
                    
                \draw[thick] (Cl1) -- (22);
                \draw[thick] (Cl2) -- (22);
                \draw[thick] (Cl3) -- (22);
    
                \node[circle,draw] (C21) at (7,6) {};
                \node[circle,draw] (C22) at (8,6) {};
    
                \draw[black] (C21) -- (C22);
    
                \draw[thick] (Cr1) -- (C21);
                \draw[thick] (Cr2) -- (C21);
                \draw[thick] (Cr3) -- (C21);
                    
                \draw[lightgray,thick] (C22) -- (Il1);
                \draw[lightgray,thick] (C22) -- (Il2);
                \draw[lightgray,thick] (C22) -- (Il3);
    
                \draw[decoration={brace,raise=70pt, amplitude=10pt},decorate] ([xshift=-2pt]21.west) -- node[above=80pt] {\huge $\frac{\lfloor\delta\rfloor+1}{2}$} ([xshift=2pt]22.east);
    
                \draw[decoration={brace, amplitude=10pt,raise=20pt},decorate] (-1,4.75) -- node[left=35pt] {\huge $\lfloor\delta\rfloor-1$} (-1,7.25);
    
                \draw[decoration={brace,raise=70pt, amplitude=10pt},decorate] ([xshift=-2pt]C21.west) -- node[above=80pt] {\huge $\lfloor\delta\rfloor-1$} ([xshift=2pt]C22.east);
            \end{tikzpicture}
        }
        \caption{Gadget for odd distances $\lfloor\delta\rfloor$ (here: $\lfloor\delta\rfloor=3$).}
        \label{fig:distanceDefensiveDomination:unattackableVertex:odd}
    \end{subfigure}
    \hfill
    \begin{subfigure}{0.48\textwidth}
        \centering
        \scalebox{0.32}{
            \begin{tikzpicture}[
                cross/.style={%
                    path picture={%
                        \draw[thick] (path picture bounding box.south west) -- (path picture bounding box.north east);
                        \draw[thick] (path picture bounding box.north west) -- (path picture bounding box.south east);
                    }
                }
            ]            
                \node[circle,draw,thick,minimum size=0.75cm] (u) at (-1,8) {}; \node[] (ul) at (-1.75,8) {\Huge $u$};
                \node[circle,draw,thick,minimum size=0.75cm] (v) at (-1,4) {}; \node[] (vl) at (-1.75,4) {\Huge $v$};
                \node[circle,draw] (1) at (-1,7) {};
                \node[circle,draw] (2) at (-1,6) {};
                \node[circle,draw] (3) at (-1,5) {};
        
                \draw[black] (u) -- (1);
                \draw[black] (v) -- (3);
                \draw[black] (1) -- (2);
                \draw[black] (2) -- (3);
        
                \node[circle,draw] (21) at (0,6) {};
                \node[circle,draw] (22) at (1,6) {};
        
                \draw[black] (2) -- (21);
                \draw[black] (21) -- (22);
    
                \node[circle,draw,minimum size=0.75cm] (Cl1) at (2.5,4.5) {};
                \node[circle,draw,minimum size=0.75cm] (Cl2) at (2.5,6) {};
                \node[circle,draw,minimum size=0.75cm] (Cl3) at (2.5,7.5) {}; \node[] (Cl) at (2.5,8.5) {\Huge $C^\ell_k$};
                \node[circle,draw,minimum size=0.75cm] (Cr1) at (5.5,4.5) {};
                \node[circle,draw,minimum size=0.75cm] (Cr2) at (5.5,6) {};
                \node[circle,draw,minimum size=0.75cm] (Cr3) at (5.5,7.5) {}; \node[] (Cr) at (5.5,8.5) {\Huge $C^r_k$};
    
                \draw[thick] (Cl1) -- (Cl2);
                \draw[thick,bend right] (Cl1) edge (Cl3);
                \draw[thick] (Cl2) -- (Cl3);
                \draw[thick] (Cr1) -- (Cr2);
                \draw[thick,bend left] (Cr1) edge (Cr3);
                \draw[thick] (Cr2) -- (Cr3);
                \draw[thick] (Cl1) -- (Cr1);
                \draw[thick] (Cl2) -- (Cr2);
                \draw[thick] (Cl3) -- (Cr3);
                \node[cross,minimum size=0.5cm] (cross1) at (4,4.5) {};
                \node[cross,minimum size=0.5cm] (cross2) at (4,6) {};
                \node[cross,minimum size=0.5cm] (cross3) at (4,7.5) {};
                
                \draw[draw=black] (11,3.5) rectangle ++(1,5); \node[] (Il) at (12.5,8) {\Huge $I_k$};
                \node[circle,draw,thick,minimum size=0.75cm] (v) at (11.5,4) {};
                \node[circle,draw,thick,minimum size=0.75cm] (v) at (11.5,5) {};
                \node[] (v) at (11.5,6) {\vdots};
                \node[circle,draw,thick,minimum size=0.75cm] (v) at (11.5,7) {};
                \node[circle,draw,thick,minimum size=0.75cm] (v) at (11.5,8) {};
                \node[] (Il1) at (11,4.75) {};
                \node[] (Il2) at (11,6) {};
                \node[] (Il3) at (11,7.25) {};
                    
                \draw[thick] (Cl1) -- (22);
                \draw[thick] (Cl2) -- (22);
                \draw[thick] (Cl3) -- (22);
    
                \node[circle,draw] (C21) at (7,6) {};
                \node[circle,draw] (C22) at (8,6) {};
                \node[circle,draw] (C23) at (9,6) {};
    
                \draw[black] (C21) -- (C22);
                \draw[black] (C22) -- (C23);
    
                \draw[thick] (Cr1) -- (C21);
                \draw[thick] (Cr2) -- (C21);
                \draw[thick] (Cr3) -- (C21);
                    
                \draw[lightgray,thick] (C23) -- (Il1);
                \draw[lightgray,thick] (C23) -- (Il2);
                \draw[lightgray,thick] (C23) -- (Il3);
    
                \draw[decoration={brace,raise=70pt, amplitude=10pt},decorate] ([xshift=-2pt]21.west) -- node[above=80pt] {\huge $\frac{\lfloor\delta\rfloor}{2}$} ([xshift=2pt]22.east);
    
                \draw[decoration={brace, amplitude=10pt,raise=20pt},decorate] (-1,4.75) -- node[left=35pt] {\huge $\lfloor\delta\rfloor-1$} (-1,7.25);
    
                \draw[decoration={brace,raise=70pt, amplitude=10pt},decorate] ([xshift=-2pt]C21.west) -- node[above=80pt] {\huge $\lfloor\delta\rfloor-1$} ([xshift=2pt]C23.east);
            \end{tikzpicture}
        }
        \caption{Gadget for even distances $\lfloor\delta\rfloor$ (here: $\lfloor\delta\rfloor=4$).}
        \label{fig:distanceDefensiveDomination:unattackableVertex:even}
    \end{subfigure}
\end{figure}
    \end{description}

    For the forward direction, we extend the defense of~\cite{DBLP:conf/mfcs/ChaplickGK25} by the $k$ tokens per gadget described above.
    The resulting defense is a set.
    Original vertices adjacent in $G$ have distance $d \leq \delta$ in $G'$, and non-adjacent ones have distance at least $2d > \delta$, so the tokens on $V(G)$ reach the same original vertices as in $G$.

    For the backward direction, as every gadget vertex is reached by the $k$ tokens of its gadget, no inclusion-minimal $k$-attack that is not countered contains a gadget vertex, and the remaining $\ell$ tokens counter every $k$-attack $A\subseteq V(G)$.
    Consider such a token at distance $x \leq d/2$ from $u$ on the subdivided path of $\{u,v\}$.
    It reaches a vertex $w \in N(v) \setminus \{u\}$ in $G$ only if $2d - x \leq \delta < d+1$, i.e., $x > d-1 \geq \frac{d}{2}$, which is impossible, and it reaches no vertex at distance at least $2$ from $u$ in $G$, as $2d \geq d+2 > \delta$.
    Tokens elsewhere in the gadget reach at most $u$ and $v$ in $V(G)$.
    Hence $u$ is a home of the token.
    \Cref{lem:homevertex} completes the proof.

    In consequence, \textsc{Defensive $\delta$-Covering With Vertex Attack} is $\Sigma_2^P$-hard for every $\delta\ge 2$.
\end{proof}

\begin{corollary}\label{cor:defensiveDeltaCoverVertexAttackMultisetDefenseHardness}
    {\sc Defensive $\delta$-Covering with Vertex Attack and Multiset Defense} is $\Sigma^P_2$-hard for any $\delta \geq 1$.
\end{corollary}
\begin{proof}
    In the cases $\delta\in[1,\frac{3}{2})$ and $\delta \geq 2$, the backward direction relies on \Cref{lem:homevertex}, which does not require the tokens to have distinct homes.
    In the case $\delta\in[\frac{3}{2},2)$, the backward direction only counts tokens per region and uses exchanges of single tokens.
    Neither argument uses that the defense is a set, and the defenses constructed in the forward directions are sets.
    Hence all reductions also apply to multiset defenses.
\end{proof}

\subsection{Multiset Attack}

We further establish complexity results concerning {\sc Defensive $\delta$-Covering With Vertex Attack And Multiset Attack}.
Since this problem is related to the $\sf NP$-complete problem {\sc $k$-tuple Dominating Set}, this problem admits different complexity behavior than its non-multiset attack variant from above.
We distinguish again between the three cases and obtain the following results:
If $\delta < \frac{1}{2}$, the problem is trivially solvable and thus in $\sf P$;
if $\frac{1}{2} \leq \delta < 1$, the problem is closely related to $b$-edge cover and thus in $\sf P$;
if $\delta \geq 1$, the problem is similar to distance-$d$ $k$-tuple dominating set and thus $\sf NP$-complete.
We start with $\delta < \frac{1}{2}$.

\begin{lemma}\label{lemma:defensiveDeltaCoverVertexAttackMultisetSmallerHalfInP}
    {\sc Defensive $\delta$-Covering With Vertex Attack And Multiset Attack} is in $\sf P$ for any $\delta < \frac{1}{2}$ regardless of whether the defense may be a multiset or not.
\end{lemma}
\begin{proof}
    Because $\delta < \frac{1}{2}$, each token of the defender is able to defend only its nearest vertex.
    First, we consider the case where the defense may be a multiset.
    Since the attack may consist of a multiset, the attacker is able to attack each vertex $k$ times.
    If a vertex is defended by fewer than $k$ tokens, the attacker will attack this very vertex.
    On the other hand, $k$ defender tokens are enough for each vertex.
    Thus, the defender needs to place exactly $k$ tokens on each vertex.
    If the defense may not be a multiset, there is a neat $(k\cdot|V(G)|)$-defense by \Cref{corollary:defensiveDeltaCoverVertexAttackNeatNiceMultiset}.
\end{proof}

\begin{theorem}\label{thm:defensiveDeltaCoverVertexAttackMultisetInP}
    {\sc Defensive $\delta$-Covering with Vertex Attack and Multiset Attack} is in $\sf P$ for any $\frac{1}{2} \leq \delta < 1$, regardless of whether the defense may be a multiset or not.
\end{theorem}
\begin{proof}
    We first consider the case where the defense may be a multiset.
    By \Cref{corollary:defensiveDeltaCoverVertexAttackNeatNice} we may restrict ourselves to considering nice $\ell$-defenses.
    Thus any defender token is always placed on an edge.
    Since the attacker can attack every vertex $k$ times, every vertex has to be defended $k$ times, as otherwise the attacker can win by attacking a single vertex that is defended at most $k-1$ times.
    Therefore {\sc Defensive $\delta$-Covering with Vertex Attack and Multiset Attack and Multiset Defense} is equivalent to choosing a multiset of edges $E^*$ (of size at most $\ell$), such that every vertex in $V(G)$ is incident to at least $k$ edges in $E^*$.
    This problem is exactly {\sc $b$-edge cover} as defined in \cite{schrijver2003combinatorial} when optimizing for the size of $E^*$ and setting all entries in the vector $b$ to $k$.
    Since {\sc $b$-edge cover} is solvable in polynomial time \cite{schrijver2003combinatorial}, {\sc Defensive $\delta$-Covering with Vertex Attack and Multiset Attack and Multiset Defense} can be solved in polynomial time as well for any $\frac{1}{2} \leq \delta < 1$.

    Now we consider the case where the defense may not be a multiset.
    For $\half < \delta < 1$, \Cref{corollary:defensiveDeltaCoverVertexAttackNeatNiceMultiset} yields that we may use the above approach as well to find the optimal solution.
    However, the case $\delta = \half$ requires a slightly more careful analysis.
    Since $\delta < 1$ and by \Cref{lemma:defensiveDeltaCoverVertexAttackFractionalPartZeroOrHalf}, we have $B^\leq(p,\frac{1}{2}) \cap V(G) = \{v\}$ for $p \in B^<(v,\half)$.
    
    If the maximum degree of the graph $G$ is at most $k$, the optimal solution is trivial:
    We place one defender token on the midpoint $p(u,v,\half)$ of each edge $\{u,v\}$, and place a number of defender tokens in $B^<(u,\half)$ such that the number of defender tokens in $\cball{u}{\half}$ is exactly $k$ for every $u \in V$.
    
    Otherwise, we define the graph $G'$:
    We copy the graph $G$, but for every vertex $v$ with degree less than $k$, we connect $k - \text{deg}(v)$ otherwise isolated vertices to $v$. 
    Each vertex that also exists in $G$ is assigned a demand of $k$, while the added vertices are assigned a demand of $1$.
    Each edge is assigned a capacity of $1$.
    We then solve the problem {\sc Capacitated $b$-Edge Cover} on $G'$, which can be done in polynomial time \cite{schrijver2003combinatorial}.
    To obtain the optimal defense, we place one defender token on $p(u,v,\half)$ for each edge $\{u,v\} \in E(G)$ in the solution to the {\sc capacitated $b$-edge cover} instance, and additionally place defender tokens in $\oball{u}{\half}$ such that the number of defender tokens in $\cball{u}{\half}$ is exactly $k$ for every $u \in V$.
    To see that this is indeed optimal, we note that any defender token can defend at most two vertices by being placed on the midpoint of an edge, since $\delta = \half$.
    For vertices with degree at most $k$ it is clear that there is an optimal solution that places a defender token on $p(u,v,\half)$ of every incident edge $\{u,v\}$ of $u$ and then additional tokens in $\oball{u}{\half}$ to defend $u$ with $k$ defender tokens.
    By solving the edge cover instance on the modified graph, we compute the lowest number of edges between high degree vertices necessary to cover the entire graph.
\end{proof}

\begin{figure}[!ht]
    \centering
    \scalebox{0.6}{
        \begin{tikzpicture}
            \node[circle,draw,minimum size=0.6cm,label=below:$v^*$] (v) at (0, 0) {};
            \node[] (vpath) at (-1, 0) {$\cdots$};
            \node[circle,draw,minimum size=0.6cm] (vpath2) at (-2, 0) {};
            \node[circle,draw,minimum size=0.6cm,label=below:$w^*$] (w) at (-3, 0) {};

            \draw[] (v) -- (vpath) -- (vpath2) -- (w);

            \draw[decoration={brace,raise=15pt, amplitude=10pt},decorate] ([xshift=-2pt]w.west) -- node[above=25pt] {$2\lfloor\delta\rfloor-1$ vertices} ([xshift=2pt]vpath.east);

            \foreach \x in {1,2,3} {
                \node[circle,draw,minimum size=0.6cm] (s\x) at (2, \x-2) {};
                \draw[] (v) -- (s\x);
            }

            \foreach \x in {1,2,3,4,5} {
                \node[circle,draw,minimum size=0.6cm] (u\x) at (4, 0.7*\x-2.1) {};
                \node[] (upath\x) at (5, 0.7*\x-2.1) {$\cdots$};
                \node[circle,draw,minimum size=0.6cm] (wu\x) at (6, 0.7*\x-2.1) {};
                \draw[] (u\x) -- (upath\x) -- (wu\x);
            }

            \draw[] (s1) -- (u1) (s1) -- (u2) (s1) -- (u4);
            \draw[] (s2) -- (u3) (s2) -- (u5);
            \draw[] (s3) -- (u1) (s3) -- (u4) (s3) -- (u5);

            \draw[decoration={brace,raise=15pt, amplitude=10pt},decorate] ([xshift=-2pt]upath5.west) -- node[above=25pt] {$\lfloor\delta\rfloor-1$ vertices} ([xshift=2pt]wu5.east);
            \draw[draw=black] (5.6,-1.8) rectangle ++(0.8,3.6); \node[] (Il) at (6.9,1) {\huge $W$};
            \draw[draw=black] (3.6,-1.8) rectangle ++(0.8,3.6); \node[] (Il) at (3.3,1.8) {\huge $U$};
            \draw[draw=black] (1.6,-1.4) rectangle ++(0.8,2.8); \node[] (Il) at (1.3,1.3) {\huge $\mathcal S$};


        \end{tikzpicture}
    }
    \caption{Sketch of the reduction graph from \Cref{thm:defensiveDeltaCoverVertexAttackMultisetBothNPComplete}.}
    \label{fig:defensiveDeltaCoverVertexAttackNPHard}
\end{figure}

\begin{theorem}\label{thm:defensiveDeltaCoverVertexAttackMultisetBothNPComplete}
    {\sc Defensive $\delta$-Covering with Vertex Attack and Multiset Attack} is $\sf NP$-complete for any $\delta \geq 1$, regardless of whether the defense may be a multiset or not.
\end{theorem}
\begin{proof}
    To show containment in \NP, we need to deal with the fact that the number of attackers and defenders may be exponential in the length of the input.
    By \Cref{lemma:defensiveDeltaCoverVertexAttackFractionalPartSmallerHalf,lemma:defensiveDeltaCoverVertexAttackFractionalPartLargerHalf,lemma:defensiveDeltaCoverVertexAttackFractionalPartZeroOrHalf}, the exact locations of the defender tokens do not matter.
    Instead we group the defender tokens into the following five groups that all can defend the same set of vertices.
    For $0 < \tilde\delta < \frac{1}{2}$, we encode the $\ell$-defense $X$ as a multiset of vertices and, for $\tilde\delta > \frac{1}{2}$, we encode the $\ell$-defense $X$ as a multiset of edges.
    By \Cref{corollary:defensiveDeltaCoverVertexAttackNeatNice}.1 and \Cref{lemma:defensiveDeltaCoverVertexAttackFractionalPartSmallerHalf} or \Cref{lemma:defensiveDeltaCoverVertexAttackFractionalPartLargerHalf}, respectively, we can obtain a neat set $\ell$-defense such that the same vertices are defended in both cases.
    In the cases of $\tilde\delta = 0$ and $\tilde\delta = \half$, it suffices to group the defender tokens as follows.
    For $\tilde\delta = 0$, we use \Cref{lemma:defensiveDeltaCoverVertexAttackFractionalPartSmallerHalf,lemma:defensiveDeltaCoverVertexAttackFractionalPartZeroOrHalf} and group them into the three sets $\{p(u,v,0)\}$, $\{p(u,v,\lambda) : \lambda \in (0,1)\}$, and $\{p(u,v,1)\}$ for each edge $\{u,v\}$.
    For $\tilde\delta = \frac{1}{2}$, we use \Cref{lemma:defensiveDeltaCoverVertexAttackFractionalPartLargerHalf,lemma:defensiveDeltaCoverVertexAttackFractionalPartZeroOrHalf} and group them into the three sets $\{p(u,v,\frac{1}{2})\}$, $\{p(u,v,\lambda) : \lambda \in [0,\frac{1}{2})\}$, and $\{p(u,v,\lambda) : \lambda \in (\frac{1}{2},1]\}$ for each edge $\{u,v\}$.
    Then, we can use \Cref{corollary:defensiveDeltaCoverVertexAttackNeatNice}.1 to obtain a neat set $\ell$-defense such that the same vertices are defended.
    Thus as a certificate, for each group represented by $X$, we encode, how many defender tokens are placed in the corresponding group, which is polynomial in the input size.

    To verify the certificate, we first compute for each $x \in X$ the set of vertices with distance at most $\delta$, which we denote by $V_x \subseteq V(G)$.
    Then we construct a directed graph $F$ with $V(F) = \{s, t\} \cup V(G) \cup X$ and $E(F)$ defined as follows:
    \begin{itemize}
        \item $s$ has an edge to each vertex in $X$ with capacity $\ell \cdot |V(G)|$
        \item each vertex $x \in X$ has an edge to each vertex in $V_x$ with capacity equal to the number of defender tokens on $x$
        \item each vertex $v \in V(G)$ has an edge to $t$ with capacity $k$.
    \end{itemize}
    Then we compute the maximum flow in $F$ and accept if its value is at least $k \cdot |V(G)|$.
    To see the correctness, note that a flow value of $k \cdot |V(G)|$ in $F$ means that there are at least $k$ defender tokens in $\cball{v}{\delta}$ for every $v \in V(G)$, ensuring that every possible attack can be defended.

    To show \NP-hardness, we give a reduction from \textsc{Set Cover}, which is very similar to the folklore reduction of \textsc{Set Cover} to \textsc{Dominating Set}.
    We first consider the case where the defense cannot be a multiset.
    Given an instance of \textsc{Set Cover} $(U, \mathcal S \subseteq 2^U, x)$ we construct an instance of {\sc Defensive $\delta$-Covering with Vertex Attack and Multiset Attack} $(G, k = 1, \ell = x+1)$.

    For $\delta \in [1,\frac{3}{2}) \cup [2, \infty)$, the graph $G$ is defined as follows:
    For each element of $U$ and $\mathcal S$ we create a vertex $v_u$ (resp. $v_S$) and connect each $v_u$ to $v_S$ if and only if $u \in S$.
    Further we add one vertex $v^*$ that is connected to each $v_S$.
    To each $v_u$ we attach a path with $\lfloor\delta\rfloor-1$ vertices, and we denote the leaf vertex of that path by $w_u$ (we say that $w_u = v_u$ if the path has length $0$).
    We denote the set of all vertices $w_u$ by $W$.
    Additionally we attach a path with $2\lfloor\delta\rfloor-1$ vertices to $v^*$, and denote the leaf vertex of that path by $w^*$.
    This can be done in polynomial time.
    An example for this construction can be seen in \Cref{fig:defensiveDeltaCoverVertexAttackNPHard}.
    
    For $\delta \in [\frac{3}{2},2)$, we use the same construction as above with the difference that we attach a path of length $2$ to $v^*$, where $w^*$ is again the leaf vertex of that path. 

    For the correctness, assume the \textsc{Set Cover} instance has a solution $S'$ of size at most $x$.
    We start with the case $\delta \in [1,\frac{3}{2}) \cup [2, \infty)$.
    We place one defender token on the $w^*$-$v^*$-path with distance $\delta$ to $w^*$ (or on $v^*$, if $1 \leq \delta < \frac{3}{2}$).
    Since $2\lfloor\delta\rfloor-1 + 2 + \lfloor\delta\rfloor-1 > 2\delta$ for all $\delta \in [1,\frac{3}{2}) \cup [2, \infty)$, this defender does not $\delta$-cover any vertex in $W$.
    On the other hand, since $2\lfloor\delta\rfloor-1 + 1 \leq 2\delta$, this defender always $\delta$-covers all vertices $v_S$.
    Additionally, we place one defender on each vertex $v_S$ with $S \in S'$.
    Since the distance of $w_u$ to $v_S$ is $\lfloor\delta\rfloor-1 + 1 \leq \delta$ if $u \in S$, these defenders together $\delta$-cover all remaining vertices of $G$ (as $\bigcup_{S \in S'}S = U$).
    Thus there is a $(x+1)$-defense on $G$ that can counter any $1$-attack.

    For the other direction, suppose there is a $(x+1)$-defense $D$ on $G$ that can counter any $1$-attack.
    Then there must be one defender token in $\cball{w^*}{\delta}$, and by the above arguments, this defender token cannot $\delta$-cover any vertex in $W$.
    W.l.o.g.\ we may thus assume that this defender token is placed on the $w^*$-$v^*$-path with distance $\delta$ to $w^*$ (or on $v^*$, if $1 \leq \delta < \frac{3}{2}$).
    Further, for any point $p$ on an edge $\{v_S,v_u\}$ or somewhere on a $v_u$-$w_u$-path, it holds that $\cball{p}{\delta} \cap W \subseteq \cball{v_S}{\delta} \cap W$.
    Since $\lfloor\delta\rfloor-1 + 2 > \delta$, $p$ only $\delta$-covers those vertices in $W$ that have distance at most $\lfloor\delta\rfloor$ from $v_S$, which are thus also $\delta$-covered by $v_S$.
    Since a set of points that together $\delta$-covers all vertices in $W$ also $\delta$-covers all other vertices in $G$ that are not $\delta$-covered by the defender token on the $w^*$-$v^*$-path, we may w.l.o.g.\ assume that all the remaining defender tokens are placed on vertices $v_S$.
    Then $S' = \{S \mid v_S \in D\}$ is a solution of size at most $x$ for the \textsc{Set Cover} instance, since a vertex $w_u$ has distance $\lfloor\delta\rfloor$ to a vertex $v_S$ if and only if $u \in S$, by construction of $G$.

    The case $\delta \in [\frac{3}{2},2)$ can be argued analogously.
\end{proof}
\section{Defensive $\delta$-Covering with Point Attack}
\label{sec:defensiveDeltaCoverWithPointAttack}

In this section we show that, surprisingly, {\sc Defensive $\delta$-Covering} remains \NP-complete even in most cases where {\sc Defensive Dominating Set} is $\Sigma^P_2$-hard.
The following two lemmas will be used as a tool to show that we can always find a defense that has a regular structure, allowing us to give an \NP-certificate with polynomial length.

\begin{lemma}[\cite{DBLP:journals/mp/HartmannLW22}]
    \label{lemma:deltaCoveringSubdivision}
    Let $\delta\text{-cover}(G)=\min\{|S|:S\text{ is a }\delta\text{-cover of }G\}$.
    For any integer $r > 1$, and graph $G$, $\delta\text{-cover}(G) = (r \cdot \delta)\text{-cover}(G_r)$.
\end{lemma}

For a rational number $\delta > 0$, we say that $\delta$ is $b$-simple, if it is equal to $a/b$ for integers $a, b$ with $b > 0$.
For a set $S \subseteq P(G)$, we say that $S$ is $b$-simple, if for every point $p = p(u_p, v_p, \lambda_p) \in S$ we have that $\lambda_p$ is $b$-simple.
The following lemma was already proven in \cite{DBLP:journals/mp/HartmannLW22}, but we give a more elaborate proof since we reuse the arguments later. 
When we only deal with defenses that are multisets, the discretization that we obtain from \Cref{lemma:deltaCoveringSimple} would already be sufficient.
However, for set defenses, we need to avoid moving different points onto the same point, which requires a more elaborate construction.
The arguments for moving the points, however, are very similar.

\begin{lemma}
    \label{lemma:deltaCoveringSimple}
    For any graph $G$ and rational $\delta = a/b$, $a,b \in \mathbb N$, there is an optimal $\delta$-cover $S^*$ that is $2b$-simple.
\end{lemma}
\begin{proof}
    First assume  $b=1$, such that $\delta$ is integer.
    For some optimal $\delta$-cover $S$, we consider the set $S^*$ constructed as follows:
    \begin{itemize}
        \item for $p = p(u,v,0) \in S$, we let $p \in S^*$
        \item for $p = p(u,v,\lambda) \in S$ with $\lambda \in (0, 1/2]$, we let $p^* = p(u,v,1/2) \in S^*$
    \end{itemize}
    Clearly, $|S^*| \leq |S|$, and since we assume $S$ to be optimal, $|S^*| = |S|$ must hold.
    Further, $S^*$ is $2$-simple.

    To show that $S^*$ is a $\delta$-cover, consider some edge $e = \{x, y\} \in E(G)$.
    We consider three cases.
    \begin{enumerate}
        \item If there is a point $p_e \in S$ with $p_e = p(u,v,0)$ that $\delta$-covers the entire edge $e$, then by construction $p_e \in S^*$ and thus $e$ is $\delta$-covered by $S^*$.
        \item If there is a point $p_e \in S$ with $p_e = p(u,v,\lambda)$ and $\lambda \in (0, 1/2]$ that $\delta$-covers the entire edge $e$, let $p_e^* = p(u,v,1/2)$ be the point that is contained in $S^*$ instead of $p_e$. 
        We consider three sub-cases.
        \begin{enumerate}
            \item $p_e^*$ is closer to $x$ and $y$ than $p_e$. 
            Then trivially $p_e^*$ also $\delta$-covers $e$.
            \item $p_e^*$ is further from $x$ and $y$ than $p_e$.
            It is easy to see that the point on $e$ furthest away from $p_e$ is one of $\{p(x,y,0), p(x,y,1/2), p(x,y,1)\}$ and that $p_e^*$ is further away from all of them than $p_e$.
            If the furthest point from $p_e$ is $p(x,y,0)$ or $p(x,y,1)$, then the distance of that point to $p_e$ is of the form $k + \lambda$ for some integer $k$.
            Since $S$ is a $\delta$-cover and $\delta$ is an integer, we have that $\delta \geq k + 1$.
            The distance of the furthest point to $p_e^*$ is $k + 1/2 < k + 1 \leq \delta$, so the edge $e$ is still $\delta$-covered by $p_e^*$.
            On the other hand, if the furthest point from $p_e$ is $p(x,y,1/2)$, the distance of that point to $p_e$ is of the form $k + 1/2 + \lambda$ for some integer $k$.
            Again since $S$ is a $\delta$-cover and $\delta$ is an integer, we have that $\delta \geq k+1 \geq k + 1/2 + \lambda$.
            Then the distance of $p_e^*$ to the furthest point is $k + 1/2 + 1/2 = k+1 \leq \delta$, and thus $e$ is $\delta$-covered by $p_e^*$.
            \item $p_e^*$ is further from $x$ and closer to $y$ than $p_e$.
            Then the two edges $\{u, v\}$ and $\{x, y\}$ are on a cycle of length at most $2\delta$, as otherwise one of the previous sub-cases would occur.
            Then any point on $\{u, v\}$ $\delta$-covers the entire edge $\{x, y\}$, and in particular also $p_e^*$.
        \end{enumerate}
        \item If there are two points $p_{e_1}, p_{e_2} \in S$ that together $\delta$-cover the entire edge $e$, let $p_{e_1}^* = p(u_1,v_1,1/2)$ and $p_{e_2}^* = p(u_2,v_2,1/2)$ be the points that are contained in $S^*$ instead of $p_{e_1}$ and $p_{e_2}$.
        Note that if there are more than two points $\delta$-covering parts of $e$, we can always find two points that together cover the entire edge, since $\delta$ is an integer.
        Further note that neither $p_{e_1}$ nor $p_{e_2}$ is at a vertex, as then we would be in case 1.
        Since neither point covers the entire edge, w.l.o.g.\ $dist(x, u_1) = \delta - 1$ and $dist(y, u_2) = \delta - 1$.
        Then $p_{e_1}^*$ $\delta$-covers every point $p(x,y,\lambda)$ for $\lambda \in [0, 1/2]$ and $p_{e_2}^*$ $\delta$-covers every point $p(y,x,\lambda)$ for $\lambda \in [0, 1/2]$.
        Thus together they $\delta$-cover the entire edge $e$.
    \end{enumerate}

    For $\delta = a/b$ with $b > 1$, we can apply the above arguments to $G_b$ and $\delta' = a$ and then use the argument of \Cref{lemma:deltaCoveringSubdivision} to obtain a $2b$-simple $\delta$-cover for $G$.
\end{proof}

Next we introduce the continuous analog to {\sc $k$-Tuple Dominating Set}, which we use as an intermediate step to show containment in \NP{} of {\sc Defensive $\delta$-Covering}.

\begin{definition}[{\sc $k$-Tuple $\delta$-Covering}]
    \label{def:kTupleDeltaCover}
    Given a graph $G$ and $k, \ell \in \mathbb N$, decide whether there is a set $S \subseteq P(G)$ of size at most $\ell$ such that for every point $p \in P(G)$ there are at least $k$ distinct points in $S$ with distance at most $\delta$ to $p$.
\end{definition}

We also consider the variant where the set $S$ is allowed to be a multiset.
In that case, we refer to the problem as {\sc $k$-Tuple Multiset $\delta$-Covering}.

\paragraph*{Defensive $\delta$-Covering with Multiset Defense}

In this section, we analyze the complexity of {\sc Defensive $\delta$-Covering with Multiset Defense}.

\begin{theorem}
    \label{thm:defensiveDeltaCoverMultisetDefenseNPComplete}
    {\sc Defensive $\delta$-Covering with Multiset Defense} is contained in $\NP$ when $\delta$ is a unit fraction, and $\NP$-complete otherwise, regardless of whether the attack may be a multiset or not.
\end{theorem}

We split the proof of this theorem into three parts:
We first show that for {\sc $k$-Tuple Multiset $\delta$-Covering} an analogous discretization argument to \Cref{lemma:deltaCoveringSimple} is possible.
Using this discretization, we then show that {\sc $k$-Tuple Multiset $\delta$-Covering} remains in $\NP$.
Finally, we show that {\sc $k$-Tuple Multiset $\delta$-Covering} and {\sc Defensive $\delta$-Covering with Multiset Defense} are equivalent.

\begin{lemma}
    \label{lemma:kTupleMultisetDeltaCover2bSimple}
    For any graph $G$ and rational $\delta = a/b$, $a,b \in \mathbb N$, there is an optimal $k$-tuple multiset $\delta$-cover that is $2b$-simple.
\end{lemma}
\begin{proof}
    Let $b = 1$ so $\delta$ is an integer, and let $S$ be an optimal $k$-tuple multiset $\delta$-cover of $G$.
    Consider some edge $e = \{x, y\} \in E(G)$.
    Then we can find $k_1$ points in $S$ that completely cover $e$, and $k_2$ pairs of points in $S$ where each pair completely covers $e$, such that $k_1 + k_2 \geq k$.

    Assume that this is not the case, i.e., there is an additional set $P_{e,m}$ of $m$ points that together cover $e$ $m' < m$ times, such that no single point in $P_{e,m}$ completely covers $e$, and removing any pair of points results in the remaining $m - 2$ points covering $e$ at most $m' - 2$ times.
    Since $\delta$ is an integer, none of the points in $P_{e,m}$ is located on $e$ or the vertices $x, y$.
    Consider two points $p_1, p_2 \in P_{e,m}$ such that $p_1$ covers the points $p(x,y,\lambda)$ for $\lambda \in [0,\mu_1]$ and $p_2$ covers the points $p(x,y,\lambda)$ for $\lambda \in [\mu_2,1]$ with $\mu_2 \leq \mu_1$ and such that $\mu_1 - \mu_2$ is minimal among all such pairs.
    Two such points must exist as otherwise $P_{e,m}$ would leave at least one point of $e$ uncovered.
    Now if $P_{e,m} \setminus \{p_1, p_2\}$ covers $e$ at most $m' - 2$ times, there must be points $p_1', p_2'$ that cover the points $p(x,y,\lambda)$ for $\lambda \in [0, \mu_1']$ and $\lambda \in [\mu_2', 1]$ respectively with $\mu_2 < \mu_1' < \mu_2' < \mu_1$, contradicting the assumption that $\mu_1 - \mu_2$ is minimal.

    Therefore we can reuse the construction of $S^*$ of \Cref{lemma:deltaCoveringSimple} and by the same arguments each of the $k_1$ points still completely covers $e$ and each of the $k_2$ pairs of points also still covers $e$.
    Further, it is easy to see that \Cref{lemma:deltaCoveringSubdivision} also applies to {\sc $k$-tuple $\delta$-Covering}, which concludes the proof.
\end{proof}

We note that it is important for the above lemma that we allow $S^*$ to be a multiset, as the construction rule of \Cref{lemma:deltaCoveringSimple} replaces all points in the interior of an edge by the same point.

\begin{lemma}
    \label{lemma:kTupleMultisetDeltaCoverNPComplete}
    {\sc $k$-Tuple Multiset $\delta$-Covering} is \NP-complete when $\delta$ is not a unit fraction, and contained in $\NP$ otherwise.
\end{lemma}
\begin{proof}
    The $\NP$-hardness for the non-unit-fraction of $\delta$ easily follows by a reduction from {\sc $\delta$-Covering} for $k=1$, which was shown to be $\NP$-hard in \cite{DBLP:journals/mp/HartmannLW22}.

    By \Cref{lemma:kTupleMultisetDeltaCover2bSimple}, there is an optimal solution that is $2b$-simple. 
    Thus for every $2b$-simple point in $P(G)$, we specify its multiplicity in $S$ as a certificate, which in total has length $\mathcal O(2b\cdot|E(G)|\cdot\log(k+1))$, which is polynomial in the input.
    To verify the certificate, we check for every point in $P(G)$ that is $4b$-simple but not $2b$-simple, whether there are at least $k$ points in $S$ with distance at most $\delta$ to it.
    If such a $4b$-simple but not $2b$-simple point $p = p(u,v,\lambda)$ is covered by some point $q$ in $S$, then $q$ must cover the entire interval $[\lambda - \frac{1}{4b}, \lambda + \frac{1}{4b}]$, since $q$ is $2b$-simple.
    Thus checking these points suffices to verify if the entire graph is covered, which can be done in polynomial time, as there are only a polynomial number of such points.
\end{proof}

\begin{lemma}
    \label{lemma:kTupleMultisetDeltaCoverEquivalence}
    {\sc Defensive $\delta$-Covering with Multiset Defense} is equivalent to {\sc $k$-Tuple Multiset $\delta$-Covering}, regardless of whether the attack may be a multiset or not.
\end{lemma}
\begin{proof}
    Let $S$ be an $\ell$-defense for any possible $k$-attack $A$ in $G$, where $A$ may be a multiset.
    If $S$ is not a $k$-tuple multiset $\delta$-cover, there is a point $p \in P(G)$ such that $\sum_{q \in \cball{p}{\delta}} \mult{q}{S} \leq k-1$.
    Then $S$ does not defend the attack consisting of $k$ tokens placed at $p$, which is a contradiction.

    On the other hand, let $S$ be a $k$-tuple multiset $\delta$-cover of size $\ell$.
    For an arbitrary $k$-attack $A$, we build the following bipartite graph $G'$ with $V(G') = V_1 \cup V_2$.
    $V_1$ contains one vertex $v_p$ for every point $p \in A$ (if $p$ is contained in $A$ multiple times, then there is a corresponding number of copies of $v_p$).
    $V_2$ contains one vertex $v_q$ for every point $q \in S$ (if $q$ is contained in $S$ multiple times, then there is a corresponding number of copies of $v_q$).
    Further, we connect vertices $v_p$ for $p \in A$ and $v_q$ for $q \in S$ if $dist(p,q) \leq \delta$ (this also holds for copies).
    If a point $q$ is contained in $S$ multiple times, multiple copies of the vertex $v_q$ are created accordingly.
    Since $S$ is a $k$-tuple multiset $\delta$-cover, each vertex $v \in V_1$ has degree at least $k$ in $G'$.
    Further, $|V_1| = k$; thus for every set $V' \subseteq V_1$ it holds that $|N(V')| \geq |V'|$ in $G'$.
    Therefore, by Hall's Theorem, there is a matching of size $k$ in $G'$, and thus $S$ counters $A$.
    Thus $S$ is also an $\ell$-defense.

    Now, we handle the case where the attack may not be a multiset.
    Let $S$ be an $\ell$-defense for any possible $k$-attack $A$ in $G$, where $A$ may not be a multiset. 
    If $S$ is not a $k$-tuple multiset $\delta$-cover, there is a point $p \in P(G)$ such that $\sum_{x \in \cball{p}{\delta}} \mult{x}{S} \leq k-1$.
    Further, let $q$ be the point in $S$ that is closest to $p$ with $dist(p,q) > \delta$.
    Thus $dist(p,q) = \delta + \varepsilon$ for some $\varepsilon > 0$.
    Let $p'$ be the point on the shortest path from $p$ to $q$ with distance $\varepsilon/2$ from $p$.
    Then by choice of $q$, for all points $p*$ on the shortest path from $p$ to $p'$ it holds that $\sum_{x \in \cball{p^*}{\delta}} \mult{x}{S} \leq k-1$.
    Since any path between two points of length greater than $0$ contains an infinite number of points, $S$ does not defend against an attack consisting of $k$ arbitrary points on the shortest path from $p$ to $p'$, which is a contradiction.

    For the other direction, it is easy to see that the above argument still works when the attack may not be a multiset.
\end{proof}

\paragraph*{Defensive $\delta$-Covering without Multiset Defense}

In this section, we analyze the complexity of {\sc Defensive $\delta$-Covering} when the defense is not allowed to be a multiset.

\begin{theorem}
    \label{thm:defensiveDeltaCoverNPCompletePolynomialK}
    {\sc Defensive $\delta$-Covering} is contained in \NP{} when $\delta$ is a unit fraction, and \NP-complete otherwise, regardless of whether the attack may be a multiset or not, as long as $k$ is polynomially bounded.
\end{theorem}

We again split the proof of the theorem into parts.
First, we provide a more elaborate discretization argument, which takes care of multiple points in the interior of an edge.
For that, we need some intermediate steps.
Throughout this, we assume that $\delta$ is an integer.

\begin{lemma}
    \label{lemma:numberOfDistinctPoints}
    If $\delta$ is an integer, then an optimal $k$-tuple $\delta$-cover contains at most $2k$ interior points of an edge.
\end{lemma}
\begin{proof}
    Let $G$ be any graph.
    Assume there is some optimal $k$-tuple $\delta$-cover $S$ that contains more than $2k$ interior points of an edge $e = \{u, v\}$ of $G$.
    Let $S_u$ be the $k$ closest points to $u$ on $e$, and let $S_v$ be the $k$ closest points to $v$ on $e$.
    It is easy to see that any point in $S$ on $e$ that is not in $S_u$ or $S_v$ only covers points that are already covered $k$ times by $S_u \cup S_v$.
    Thus $S$ is still a $k$-tuple $\delta$-cover without these points, which is a contradiction to $S$ being optimal.
\end{proof}

For a graph $G$, we consider an arbitrary optimal $k$-tuple $\delta$-cover $S$.
Let $G_S$ be the graph with vertex set $S$ and edge set $E(G_S) = \{\{p, q\} ~|~ dist(p,q) = 2\delta \}$.
The set of points corresponding to a connected component of $G_S$ is an \emph{equivalence class}.
Thus for each edge, at most two points in $S$ belong to the same equivalence class.
The set of all equivalence classes of a $k$-tuple $\delta$-cover $S$ is denoted by $EQ(S)$.
According to \Cref{lemma:numberOfDistinctPoints}, $|EQ(S)| \leq 2k\cdot|E(G)|$ for any graph $G$ and optimal $k$-tuple $\delta$-cover $S$ of $G$.

Thus the goal is to modify the construction of $S^*$ from \Cref{lemma:deltaCoveringSimple} as follows:
\begin{itemize}
    \item for $p = p(u,v,0) \in S$ or $p = p(u,v,\half) \in S$, we let $p \in S^*$
    \item for $p = p(u,v,\lambda) \in S_i$ with $0 < \lambda < \half$, we let $p^* = p(u,v,1/2 - \pi(i)/(4k\cdot|E(G)|)) \in S^*$, where $\pi$ is a permutation of length $|EQ(S)|$.
\end{itemize}
In the following, we show that for any $S$ a permutation $\pi$ exists such that $S^*$ is a valid $k$-tuple $\delta$-cover of $G$.

For an equivalence class $S_i$ we denote by $\lambda_i$ the shortest distance to a vertex of all points in $S_i$.
Two equivalence classes $S_i, S_j$ can impose a \emph{restriction} of the form $\lambda_i \leq \lambda_j$, if there are points $p = p(u,v,\lambda_i) \in S_i$ and $q = p(w,x,\lambda_j) \in S_j$ (w.l.o.g.\ let $\lambda_i < 1/2$ and $\lambda_j < 1/2$) meeting the following requirements:
\begin{enumerate}
    \item $2\delta-1 < dist(p,q) < 2\delta$
    \item $\lambda_i \leq \lambda_j$
    \item there are $\lambda'_i, \lambda'_j$ with $\half > \lambda'_i > \lambda'_j$ such that $dist(p(u,v,\lambda'_i),p(w,x,\lambda'_j)) > 2\delta$
\end{enumerate}
Note that the equivalence classes that contain the points with $\lambda = 0$ and $\lambda=\half$ never impose restrictions on other equivalence classes due to condition 3.

\begin{lemma}
    \label{lemma:noCircularRestriction}
    There are no equivalence classes $S_{i_1}, \dots, S_{i_n}$ imposing restrictions of the form $\lambda_{i_1} \leq \dots \leq \lambda_{i_n} \leq \lambda_{i_1}$.
\end{lemma}
\begin{proof}
    From the second requirement it follows that in this case it would hold that $\lambda_{i_1} = \dots = \lambda_{i_n}$.
    Thus the distance of two points from these equivalence classes is either an integer, violating the first requirement, or of the form $2\delta - 2\lambda_{i_1}$, violating the third requirement, as $dist(p(u,v,1/2),p(w,x,1/2)) \leq 2\delta$ and $dist(u,w) \leq 2\delta$, which follows from $\delta$ being an integer.
\end{proof}

\begin{lemma}
    \label{lemma:permutationFulfillingAllRestrictions}
    There is a permutation $\pi$ such that if $\pi(i) < \pi(j)$ then there is no restriction of the form $\lambda_{i} > \lambda_{j}$.
\end{lemma}
\begin{proof}
    We create a directed graph, where each equivalence class is represented by a vertex, and each restriction of the form $\lambda_i \leq \lambda_j$ is modeled via a directed edge from the vertex representing $S_i$ to the vertex representing $S_j$.
    By \Cref{lemma:noCircularRestriction}, this graph contains no directed cycle.
    Since we can embed a directed acyclic graph into a directed path without reversing an edge, the claim follows.
\end{proof}

From \Cref{lemma:numberOfDistinctPoints,lemma:permutationFulfillingAllRestrictions} it follows that the modified construction of $S^*$ above does not discretize different points in an optimal $k$-tuple $\delta$-cover $S$ to the same point, as the construction of \Cref{lemma:deltaCoveringSimple} does.
We now look at the three cases in the proof of \Cref{lemma:deltaCoveringSimple}:
The first two cases still hold by the same arguments, and the third case holds due to \Cref{lemma:permutationFulfillingAllRestrictions}.
Combining this with the argument of \Cref{lemma:kTupleMultisetDeltaCover2bSimple} to consider the points (resp. pairs of points) that cover an edge separately, we obtain the following theorem.

\begin{theorem}
    \label{theorem:kTupleDeltaCoverSimple}
    For any graph $G$ and rational $\delta = a/b$, $a,b \in \mathbb N$, there is an optimal $k$-tuple $\delta$-cover that is $\left(4kb \cdot |E(G)|\right)$-simple.
\end{theorem}

Note that \Cref{lemma:deltaCoveringSubdivision} can easily be extended to $k$-tuple $\delta$-cover to handle the cases where $\delta$ is not an integer.

As long as $k$ is polynomially bounded in the size of the input graph, we can infer from \Cref{theorem:kTupleDeltaCoverSimple} that {\sc $k$-Tuple $\delta$-Covering} is contained in NP, by the same argument as in the proof of \Cref{lemma:kTupleMultisetDeltaCoverNPComplete}.
However, since the length of the certificate is $\mathcal O(4kb \cdot |E(G)|^2 \cdot \log(k+1))$, it is no longer polynomially bounded in the input when $k$ is not.
Next, we show that {\sc Defensive $\delta$-Covering} is equivalent to {\sc $k$-Tuple $\delta$-Covering}, completing the proof that {\sc Defensive $\delta$-Covering} is contained in \NP{} as long as the size of the attack is bounded by a polynomial.

\begin{lemma}
    \label{lemma:kTupleDeltaCoverEquivalence}
    {\sc Defensive $\delta$-Covering} is equivalent to {\sc $k$-Tuple $\delta$-Covering}, regardless of whether the attack may be a multiset or not.
\end{lemma}
\begin{proof}
    This follows directly from the arguments of \Cref{lemma:kTupleMultisetDeltaCoverEquivalence}, as all arguments in that proof are independent of the defense being a multiset.
\end{proof}

\section{Conclusion and Future Work}
\label{sec:conclusion}

We thoroughly analyzed the complexity of {\sc Defensive Domination} in a continuous setting, showing that the continuous setting often makes the problem easier from a complexity standpoint.
This leaves open how the minimum size of a defense that counters every $k$-attack compares between the vertex-attack and point attack variants.
Additionally, in \cite{DBLP:journals/mp/HartmannLW22} the authors provide a polynomial time algorithm for {\sc $1$-Covering}.
This algorithm easily extends to {\sc $k$-Tuple $\delta$-Covering} on bipartite graphs, as the set of factor-critical components of the Edmonds-Gallai decomposition of size at least $3$ is empty for bipartite graphs.
It would be interesting to determine whether this algorithm as well as the shifting theorem from \cite{DBLP:journals/mp/HartmannLW22} extend to {\sc $k$-Tuple $\delta$-Covering}, as this would settle the cases left open by \Cref{tab:my_label}. 

Another problem variant one could consider would be {\sc Defensive $\delta$-Covering with Vertex Defense}, where only the attacker is allowed to place tokens on any point in the graph, while the defender is restricted to vertices.
However, we suspect that this problem behaves very similarly to {\sc Defensive $\delta$-Covering}, since the attacker again can implicitly attack a point multiple times, by placing tokens arbitrarily close together.
\newpage
\bibliography{bibliography}

\appendix

\end{document}